\documentclass[journal]{IEEEtran}%\documentclass[a4aper,11pt,onecolumn,final]{paper}
\usepackage{amsmath,amsfonts,pstricks,graphicx}
\usepackage{amssymb,xcolor,colortbl}
\usepackage{rotating}
\usepackage{framed}

\input xy
\xyoption{all}

\usepackage{nopageno}
\usepackage{bm}
\usepackage{psfrag,fancybox}
\usepackage{framed}
\newcounter{fdefncount}
\definecolor{shadecolor}{rgb}{1,1,1}
\newtheorem{niall1}{Theorem}[section]
\newtheorem{niall2}{Corollary}[section]
\newtheorem{niall3}{Lemma}[section]
\newtheorem{niall4}{Definition}[section]
\newtheorem{niall5}{Experiment}[section]
\newtheorem{niall6}{Algorithm}[section]
\newtheorem{niall7}{Hypothesis}[section]
\newtheorem{niall8}{Example}[section]
\newtheorem{niall10}[fdefncount]{Proof}
\newtheorem{niall9}{Proof}[section]
\newtheorem{niall11}{Counter~Example}[section]
\newenvironment{thm}{\begin{niall1}}{\end{niall1}}

\newenvironment{counterex}{\begin{niall11}}{\end{niall11}}

\newcommand{\bmat}{\left[ \begin{array}}
\newcommand{\emat}{\end{array} \right]}
\newcommand{\bpmat}{\begin{pmatrix}}
\newcommand{\epmat}{\end{pmatrix} ]}
\newcommand{\bbmat}{\begin{bmatrix}}
\newcommand{\ebmat}{\end{bmatrix} ]}

\newcommand{\vc}{{\vec{c}}}

\newcommand{\vd}{{\vec{d}}}
\newcommand{\ve}{{\vec{e}}}
\newcommand{\vC}{{\vec{C}}}
\newcommand{\vD}{{\vec{D}}}
\newcommand{\vE}{{\vec{E}}}
\newcommand{\vF}{{\vec{F}}}

\newcommand{\bbR}{\mathbb{R}}
\newcommand{\bbC}{\mathbb{C}}
\newcommand{\bbN}{\mathbb{N}}

\newcommand{\bbox}{\definecolor{shadecolor}{rgb}{0.8,0.8,0.8}\begin{shaded}}
\newcommand{\ebox}{\end{shaded}\definecolor{shadecolor}{rgb}{0.9,0.9,0.9}}

\newcommand{\degree}{^{\circ}}

\def\vc{{\vec{c}}}
\def\Da{{$D1$}}
\def\Db{{$D2$}}
\def\Dc{{$D3$}}
\def\Dd{{$D4$}}
\def\De{{$P1$}}
\def\Df{{$P2$}}
\def\Aa{{$A1$}}
\def\Ab{{$A2$}}

\newcommand{\sa}{{s_1}}
\newcommand{\saa}{{s_2}}
\def\beq{\begin{eqnarray*}}
\def\eeq{\end{eqnarray*}}

\newcommand{\half}{\frac{1}{2}}

\renewcommand{\thefigure}{\arabic{figure}}
\title{Comparing Measures of Sparsity}
\author{Niall~Hurley~and~Scott~Rickard\thanks{N. Hurley and S. Rickard are with the  Sparse Signal Processing Group in, University College Dublin, Ireland
This material is based upon works supported by the Science Foundation
Ireland under Grant No. 05/YI2/I677.}}

\begin{document}

\maketitle

%\pagenumbering{arabic}
%\pagestyle{headings}

\begin{abstract}
Sparsity of representations of signals has been shown to be a key
concept of fundamental importance in fields such as blind source
separation, compression, sampling and signal analysis.  The aim of
this paper is to compare several commonly-used sparsity measures based
on intuitive attributes.  Intuitively, a sparse representation is one
in which a small number of coefficients contain a large proportion of
the energy.  In this paper six properties are discussed: ({\em Robin
Hood}, {\em Scaling}, {\em Rising Tide}, {\em Cloning}, {\em Bill
Gates} and {\em Babies}), each of which a sparsity measure should
have. The main contributions of this paper are the proofs and the
associated summary table which classify commonly-used sparsity
measures based on whether or not they satisfy these six propositions.  Only one of these measures satisfies
all six: the Gini Index.
\end{abstract}

\section{Introduction}
\label{sec:intro} Whether with sparsity constraints or with sparsity
assumptions, the concept of sparsity is readily used in diverse areas
such as oceanic engineering \cite{li07estimation}, antennas and
propagation \cite{adams08sparse}, face recognition
\cite{wright08robust}, image processing
\cite{mairal08sparse,aharon08sparse} and medical imaging
\cite{leung08sparse}. Sparsity has also played a central role in the
success of many machine learning algorithms and techniques such as
matrix factorization \cite{gupta97highly}, signal recovery/extraction
\cite{ tropp06just}, denoising \cite{candes07enhancing, zhu08stable},
compressed sensing \cite{mishali08reduce}, dictionary learning
\cite{aharon06algorithm}, signal representation \cite{akcakaya08frame,
gribonval03sparse}, support vector machines \cite{licheng07sparse},
sampling theory \cite{blu08sparse,goyal08compressive} and source
separation/localization \cite{ogrady05survey,he07convolutive}. For
example, one method of source separation is to transform the signal to
a domain in which it is sparse (e.g.  time-frequency or wavelet) where
the separation can be performed by a partition of the transformed
signal space due to the sparsity of the representation
\cite{rickard07DUET, yilmaz04blind}. There has also been research in
the uniqueness of sparse solutions in overcomplete representations
\cite{elad06sparse,bruckstein08uniqueness}.

There are many measures of sparsity. Intuitively, a sparse
representation is one in which a small number of coefficients contain
a large proportion of the energy. This interpretation leads to further
possible alternative measures. Indeed, there are dozens of measures of
sparsity used in the literature. Which of the sparsity measures is the
best?  In this paper we suggest six desirable characteristics of
measures of sparsity and use them to compare fifteen popular sparsity
measures.

Considering the nebulous definition of sparsity we begin by examining
how a sparsity measure should behave in certain scenarios.  In
Sec.~\ref{sec:crit} we define six such scenarios and formalize these
scenarios in six mathematical criteria that capture this desirable
behavior. We prove two theorems showing that satisfaction of some
combinations of criteria result in automatic compliance with a
different criteria. In Sec.~\ref{sec:meas} we introduce the most
commonly-used sparsity measures in the literature. We elaborate on one
of these measures, the Gini Index, as it has many desirable
characteristics including the ability to measure the sparsity of a
distribution. We also show graphically how some measures treat
components of different magnitude.  In Sec.~\ref{sec:comp} we present
the main result of this work, namely, the comparison of the fifteen
commonly-used sparsity measures using the six criteria. We show that
the only measure to satisfy all six is the Gini Index.  Proofs of the
table are attached in Appendices~\ref{sec:cex} and \ref{sec:proofs}. A
preliminary report on these results (without proofs) appeared in
\cite{hurley08comparing}. We then compare the fifteen measures
graphically on data drawn from two sets of parameterized
distributions. We select distributions for which we can control the
`sparsity'. This allows us to visualize the behavior of the sparsity
measures in view of the sparse criteria. In Sec.~\ref{sec:concl} we
present some conclusions. The main conclusion is that from the fifteen
measures, only the Gini Index satisfies all six criteria, and, as
such, we encourage its use and study.

\section{The Six Criteria}
\label{sec:crit}
The following are six desirable attributes of a measure of
sparsity. The first four, \Da~through \Dd, were originally applied in
a financial setting to measure the inequity of wealth distribution in
\cite{dalton20measurement}. The last two, \De~and \Df, were proposed
in \cite{rickard04gini}. Distribution of wealth can be used
interchangeably with distribution of energy of coefficients and where
convenient in this paper, we will keep the financial interpretation in
the explanations.  Inequity of distribution is the same as
sparsity. An equitable distribution is one with all coefficients
having the same amount of energy, the least sparse distribution.
%\cite{rickard04gini,arnold86majorization,dalton20measurement,gini21measurement,lorenz05methods},

\begin{itemize}
\item[\Da] {\em Robin Hood} - Robin Hood decreases sparsity (Dalton's
1st Law). Stealing from the rich and giving to the poor decreases the
inequity of wealth distribution (assuming we do not make the rich poor
and the poor rich). This comes directly from the definition of a
sparse distribution being one for which most of the energy is
contained in only a few of the coefficients.
\item[\Db] {\em Scaling} - Sparsity is scale invariant (Dalton's
modified 2nd Law \cite{arnold86majorization}). Multiplying wealth by a
constant factor does not alter the effective wealth distribution. This
means that relative wealth is important, not absolute wealth. Making
everyone ten times more wealthy does not affect the effective
distribution of wealth. The rich are still just as rich and the poor
are still just as poor.
\item[\Dc] {\em Rising Tide} - Adding a constant to each coefficient
decreases sparsity (Dalton's 3rd Law).  Give everyone a trillion
dollars and the small differences in overall wealth are then
negligible so everyone will have effectively the same wealth. This is
intuitive as  adding a constant energy to each coefficient reduces
the relative difference of energy between large and small
coefficients. This law assumes that the original distribution contains
at least two individuals with different wealth. If all individuals
have identical wealth, then by \Db~ there should be no change to the
sparsity for multiplicative or additive constants.
\item[\Dd] {\em Cloning} - Sparsity is invariant under cloning
(Dalton's 4th Law). If there is a twin population with identical
wealth distribution, the sparsity of wealth in one population is the
same for the combination of the two.
\item[\De] {\em Bill Gates} - Bill Gates increases sparsity. As one
individual becomes infinitely wealthy, the wealth distribution becomes
as sparse as possible.
\item[\Df] {\em Babies} - Babies increase sparsity. In populations
with non-zero total wealth, adding individuals with zero wealth to a
population increases the sparseness of the distribution of wealth.
\end{itemize}

%\Da~comes directly from the definition of a sparse distribution being
%one for which most of the energy is contained in only a few of the
%coefficients. Alternatively stated, a sparse distribution is one for
%which most of the wealth is in the hands of a few. \Db~ means that relative wealth is important, not absolute wealth. Making everyone ten times more wealthy does not affect the effective distribution of wealth. The rich are still just as rich and the poor are still just as poor.   \Dc~is
%intuitive as by adding a constant energy to each coefficient reduces
%the relative difference of energy between the larger coefficients and
%the smaller.

%%------------------------------------------------------%%
%%------------------------------------------------------%%
%%------------------------------------------------------%%

These criteria give rise to the sparsest distribution being one with
one individual owning all the wealth and the least sparse being one
with everyone having equal wealth.

Dalton \cite{dalton20measurement} proposed that multiplication by a
constant should decrease inequality. This was revised to the more
desirably property of scale invariance. Dalton's fourth principle,
\Dd, is somewhat controversial. However, if we have a distribution
from which we draw coefficients and measure the sparsity of the
coefficients which we have drawn, as we draw more and more
coefficients we would expect our measure of sparsity to converge.
\Dd~ captures this concept.
\begin{quote}
`Mathematically this [\Dd] requires that the measure of inequality of
the population should be a function of the sample distribution
function of the population. Most common measures of inequality
satisfy this last principle.'\cite{arnold86majorization}
\end{quote}
Interestingly, most measures of sparsity do not satisfy this
principle, as we shall see.

We define a sparse measure $S$ as the a function with the following
mapping
\begin{equation}\label{eq:spmeasC}
S:\left(\bigcup_{n\geq1}\bbC^n \right)\rightarrow\bbR
\end{equation}
where $n \in \bbN$ is the number of coefficients. Thus $S$ maps
complex vectors to a real number.

There are two crucial, core, underlying attributes which our sparsity
measures must satisfy. As all measures satisfy these two conditions
trivially we will not comment on them further except to define them.
\begin{itemize}
\item[\Aa] $S(\vc) = S(\Pi\vc)$ where $\Pi$ denotes permutation, that
is, the sparsity of any permutation of the coefficients is the
same. This means that the ordering of the coefficients is not
important.
\item[\Ab] The sparsity of the coefficients is calculated using the magnitudes of the coefficients. This
means we can assume we are operating in the positive orthant, without
loss of generality.
%% $S(\vc) = S(|\vc|)$ where $|\cdot|$ denotes element-wise
%% absolute value, that is, the sparsity of the absolute value of the
%% coefficients is the same as the sparsity of the coefficients. This
%% means we can assume we are operating in the positive orthant, without
%% loss of generality.
\end{itemize}

By \Ab~ we can assume we are operating in the positive orthant, and as
such we can rewrite (\ref{eq:spmeasC}) as
\begin{equation}\label{eq:spmeasR}
S:\left(\bigcup_{n\geq1}\bbR_{+}^n \right)\rightarrow\bbR,
\end{equation}
which is more consistent with the wealth interpretation.

We will use the convention that $S(\vc)$ increases with increasing
sparsity where $\vc=\bmat{cccc} c_1& c_2&\cdots \emat$ are the
coefficient strengths. Given vectors
\begin{eqnarray*}
\vc&=&\bmat{ccccc} c_1& c_2&\cdots &c_N\emat\\
\vd&=&\bmat{ccccc} d_1& d_2&\cdots &d_M\emat
\end{eqnarray*}
we define concatenation, which we use $\|$ to denote, as
\[
\vc\|\vd = \bmat{cccccccc} c_1& c_2&\cdots &c_N&d_1& d_2&\cdots &d_M\emat.
\]
We also define the addition of adding a constant to a vector as the
addition of that constant to each element of the vector, that is,
for $\alpha in \bbR$,
\[
\vc + \alpha  =\bmat{ccccc} c_1 + \alpha& c_2 + \alpha&\cdots &c_N + \alpha\emat.
\]
The six sparse criteria can be formally defined as follows:
\begin{itemize}
\item[\Da] {\em Robin~Hood}:\\{\small $S(\bmat{cccccc}c_1& \cdots& c_i-\alpha&
\ldots& c_j+\alpha& \ldots\emat )$ $<S(\vc)$} for all $\alpha, c_i, c_j $
such that $c_i>c_j$ and $0<\alpha < \frac{c_i-c_j}{2}$.
\item[\Db] {\em Scaling}:\\ $S(\alpha\vc) = S(\vc)$, $\forall \alpha \in \mathbb{R},~\alpha > 0$.
\item[\Dc] {\em Rising Tide}:\\$S(\alpha + \vc)<S(\vc)$, $\alpha \in
\mathbb{R},~\alpha > 0$ (We exclude the case $c_1 = c_2 = c_3 = \cdots
= c_i = \cdots \forall i$ as this is equivalent to scaling.).
\item[\Dd] {\em Cloning}:\\$ S(\vc) = S(\vc\|\vc) = S(\vc\|\vc\|\vc) = S(\vc\|\vc\|\cdots\|\vc)$.
\item[\De] {\em Bill Gates}:\\  $\forall i \exists \beta = \beta_i>0$, such that $\forall \alpha >0:$ {\small \[S(\bmat{cccc}c_1& \ldots& c_i+\beta+\alpha&\ldots\emat)>S(\bmat{cccc}c_1&\ldots&c_i+\beta& \ldots\emat).\]}
\item[\Df] {\em Babies}:\\$S(\vc||0) > S(\vc)$.
\end{itemize}

As stated above, when proving {\em Rising Tide} we exclude the
scenario where all coefficients are equal. In this case, adding a
constant is actually a form of scaling. Another interpretation is that
the case with all coefficients equal is, in fact, the minimally sparse
scenario and hence adding a constant cannot decrease the sparsity.

\subsection{Two Proofs}
As one would surmise there is some overlap between the criteria. We
present and prove two theorems which demonstrate this overlap.
Theorem~\ref{thm:1} states that if a measure satisfies both criteria
\Da~ and \Db, the sparsity measure also satisfies \De~ by default.
Theorem~\ref{thm:2} states that a measure satisfying \Da, \Db~ and
\Dd~ necessarily satisfies \Df.
%___________________________________________________________________________%
\begin{thm}{\em $D1~\&~D2\Rightarrow P1$, that is, if both \Da~ and \Db~
  are satisfied, \De~ is also satisfied.}
\label{thm:1}
\end{thm}
\begin{proof}
  Without loss of generality, we begin with the vector $\vc$ sorted in
  ascending order
\begin{equation}
  \nonumber \vc=\bmat{cccc} c_1& c_2&\cdots&c_N \emat
\end{equation}
with $c_{1}\leq c_{2} \leq\cdots\leq c_{N}$. We then perform a
series of inverse Robin Hood steps to get a vector $\vd$, that is, we
take from smaller coefficients and give to the largest coefficient
\begin{eqnarray*}
  d_i  &=& c_i - \Delta c_i  ~~~~~~~~ \forall i = 1,2,\ldots,N-1\\
 d_N &=& c_N + \Delta c_i
 \end{eqnarray*}
with condition $\Delta<1$. As these are inverse Robin Hood steps
(inverse \Da), they increase sparsity and result in the vector
{\footnotesize
\begin{equation}
 \nonumber \vd = \begin{array}{ccl} \left[ (c_1-\Delta c_1)\right.& (c_2-\Delta c_2)&\cdots\\
                             \cdots&{\footnotesize (c_{N-1}-\Delta c_{N-1})}&\left.{\tiny (\Delta c_1+\cdots+\Delta c_{N-1} +c_N)}\right]\end{array}.
\end{equation}}
Without affecting the sparsity we can then scale (\Db) $\vd$ by $\frac{1}{1-\Delta}$ to get
{\small \begin{eqnarray*}
 \ve &=& \begin{array}{ccl}\left[ c_1\right.& c_2&\cdots\\\cdots&c_{N-1}&\left.{\footnotesize\frac{1}{1-\Delta}( \Delta c_1+\Delta c_2+\ldots+\Delta c_{N-1} +c_N)}\right]\end{array}\\
 &=&\begin{array}{ccl}\left[ c_1\right.& c_2&\cdots\\\cdots&c_{N-1}&\left.\alpha + c_N\right],\end{array}
\end{eqnarray*}}
where
\begin{eqnarray*}
  \alpha &=& \frac{1}{1-\Delta}(\Delta c_1+\Delta c_2+\ldots+\Delta c_{N}).
\end{eqnarray*}
It is clear that
\begin{equation}
  \nonumber S(\ve) = S(\vd)>S(\vc),
\end{equation}
which is equivalent to \De~ with the given $\alpha$ and $\beta=0$.  If we wish to operate on $c_i$ (instead of $c_N$ as
above), $\beta$ can be chosen sufficiently large to make the desired coefficient the
largest, that is, we set
\[
\beta > c_N - c_i
\]
\end{proof}
\begin{thm}
{\em $D1~\&~D2~\&~D4\Rightarrow P2$, that is, if  \Da, \Db~ and \Dd~
 are satisfied, \Df~ is also satisfied.}\label{thm:2}
\end{thm}
\begin{proof}
We begin with vector $\vc$
\begin{equation}
  \nonumber \vc=\bmat{cccc} c_1& c_2&\ldots&c_N \emat.
\end{equation}
We then clone (\Dd) this $N+1$ times to get
\begin{equation}
\nonumber
 \begin{array}{cccc}
  \vC&=& \left[ \underbrace{\begin{array}{cccc} \vc& \vc&\ldots&\vc \end{array}}\right].& \\
&&N+1&
  \end{array}
\end{equation}
We then take one of the $\vc$ from $\vC$, which we shall refer to as
$\vec{\hat{c}}$ and by a series of inverse Robin Hood operations (\Da)
we distribute this $\vec{\hat{c}}$ in accordance with the size of each element
to form new vector $\vD$. That is to say, each $c_i$ of each $\vc$
(excluding $\vec{\hat{c}}$) becomes $c_i +\frac{c_i}{N}$ by $N$ consecutive inverse
Robin Hood operations which increase sparsity. The result is
{\small
\begin{equation}
  \nonumber \begin{array}{ccccc}
  \vD &=&\left[\underbrace{ \begin{array}{cccc}\vc +\frac{\vc}{N}& \vc +\frac{\vc}{N}&\cdots&\vc +\frac{\vc}{N} \end{array}}\right.&\left.\underbrace{\begin{array}{cccc}0&0&\cdots&0\end{array}}\right].& \\
&&$N$&$N$&
  \end{array}
\end{equation}}
We can then scale (\Db) $\vD$ by a factor of $\frac{N}{1+N}$ without
affecting the sparsity to get
\begin{equation}
  \nonumber \begin{array}{ccccc}
  \vE &=&\left[\underbrace{ \begin{array}{cccc}\vc& \vc&\cdots&\vc \end{array}}\right.&\left.\underbrace{\begin{array}{cccc}0&0&\cdots&0\end{array}}\right],& \\
&&$N$&$N$&
  \end{array}
\end{equation}
which by cloning (\Dd) we know is equivalent to
\begin{equation}
  \nonumber \vF = \bmat{cc}\vc &0\emat.
\end{equation}
In summation, we have shown that
\begin{equation}
 S(\vc) =  S(\vC)< S(\vD) =  S(\vE) = S(\vF),
\end{equation}
that is,
\begin{equation}
 S(\vc) < S(\vc\|0)
\end{equation}
which is also known as \Df.
\end{proof}
%___________________________________________________________________________%
\section{The Measures of Sparsity}
\label{sec:meas}
In this section we discuss a number of popular sparsity
measures. These measures are used to calculate a number which
describes the sparsity of a vector $ \vc=\bmat{cccc} c_1&
c_2&\ldots&c_N \emat$. The measures' monikers and their definitions
are listed in Table~\ref{table:measurerefs}.  Some measures in
Table~\ref{table:measurerefs} have been manipulated (in general
negated) to ensure that the an increase in sparsity results in a
(positive) increase in the sparse measure.
%\cite{karvanen03measuring,rao99affine}

%___________________________________________________________________________%
%\hspace*{-.5 in}
\begin{table}[ht]
\caption{Commonly used sparsity measures modified to become more positive for increasing sparsity.}
\begin{center}
\begin{tabular}{|c|l|}
\hline
Measure &  Definition \\
\hline
\hline
%measure
$\ell^0$ & $\#\left\{j,c_j = 0\right\}$ \\
\hline
$\ell^0_\epsilon$ & $\#\left\{j,c_j\leq \epsilon\right\}$ \\
\hline
$-{\ell^1}$ & $-{\left(\sum_j c_j \right)}$ \\
\hline
$-{\ell^p}$ & $-{\left(\sum_j c_j^p\right)^{1/p},\,\,\,0<p<1}$  \\
\hline
$\frac{\ell^2}{\ell^1}$ &$\frac{\sqrt{\sum_j c_j^2}}{\sum_j c_j}$ \\
\hline
$-{\tanh_{a,b}}$  & $-{\sum_j\tanh\left(\left(ac_j\right)^b\right)}$  \\
\hline
  $-\log$ & $-\sum_j\log\left(1+c_j^2\right)$ \\
\hline
  $\kappa_4$ & $\frac{\sum_j c_j^4}{\left(\sum_j c_j^2\right)^2}$ \\
\hline
   ~   &    $1-\min_{i=1,2,\ldots,N-\lceil \theta N\rceil+1}\frac{c_{(i+\lceil \theta N\rceil -1)}-c_{(i)}}{c_{(N)} -c_{(1)}|} $        \\
      $u_\theta$        &  $\text{ s.t. } \lceil \theta N \rceil \neq N$  $\text{ for ordered data, }$     \\
    ~             & $c_{(1)}\leq c_{(2)} \leq\cdots\leq c_{(N)}$ \\
\hline
  $-\ell^p_{-}$ & $-\sum_{j,c_j\neq 0}c_j^p,\,\,\, p<0$ \\
\hline
  $H_G$  & $-\sum_j\log  c_j ^2$ \\
\hline
 $H_S$ & $-\sum_j \tilde{c_j}\log  \tilde{c_j} ^2 \text{ where }
\tilde{c_j}=\frac{c_j^2}{\|\vc\|_2^2}$ \\
\hline
 $H_S^\prime$ & $-\sum_j  {c_j}\log   {c_j} ^2$ \\
\hline
 Hoyer & $(\sqrt{N}-\frac{\sum_j c_j}{\sqrt{\sum_j c_j^2}})(\sqrt{N}-1)^{-1}$ \\
\hline
  ~   &   $1 - 2\sum_{k=1}^N \frac{c_{(k)}}{\|\vc\|_1}\left(\frac{N-k+\frac{1}{2}}{N}  \right) $\\
      Gini        &    $\text{ for ordered data, }$     \\
    ~             & $c_{(1)}\leq c_{(2)} \leq\cdots\leq c_{(N)}$ \\
\hline
\end{tabular}
\label{table:measurerefs}
\end{center}
\end{table}

In \cite{karvanen03measuring} the $\ell^0$, $\ell_{\epsilon}^0$,
$\ell^1$, $\ell^p$, $\tanh_{a,b}$, $\log$ and $\kappa_4$ were
compared.  The most commonly used and studied sparsity measures are
the $\ell^p$ norm-like measures,
\[
\|\vc\|_p = \left( \sum_j c_j^p \right)^{1/p} ~~~\mbox{for~} 0\leq p \leq 1.
\]
%% When $0<p<1$, $\ell^p$ is a quasi-norm, as it does not satisfy the
%% triangle inequality.  For $p=0$, $\ell^p$ is not linear with respect
%% to scalar multiplication and hence is not a norm or even a
%% quasi-norm. We shall refer to this as the $\ell^0$ measure. We refer
%% to all of the measures as measures or costs to avoid confusion with
%% norm, quasi-norm,and so on.
The $\ell^0$ measure simply calculates
the number of non-zero coefficients in $\vc$,
\[
\|\vc\|_0 = \# \{ c_j\neq 0, j = 1, \ldots,N\}.
\]
 The $\ell^0$ measure is the traditional sparsity measure in many
mathematical settings. However, it is unsuited to most practical
scenarios, as an infinitesimally small value is treated the same as a
large value. This means that the derivative of the measure contains no
information and as such the $\ell^0$ cannot be used in optimization
problems. Exhaustive search is the only method of finding the sparsest
solution when using the $\ell^0$ measure and approximations are
usually used \cite{fuchs05recovery, donoho06stable}.  The presence of
noise makes the $\ell^0$ measure completely inappropriate. In noisy
settings, the $\ell^0$ measure is sometimes modified to
$\ell^0_{\epsilon}$ where we are interested in the number of
coefficients, $c_j$ that are greater than a threshold $\epsilon$
\cite{rath08sparse}.  Clearly, the value of $\epsilon$ is crucial for
$\ell^0_\epsilon$ to be meaningful. This is undesirable.  As
optimization using $\ell^0_{\epsilon}$ is difficult because the
gradient yields no information, $\ell^{p}$ with $0<p<1$ is often used
in its place, \cite{xu07norm}.  The $\ell^1$ measure, that is,
$\ell^{p}$ with $p=1$, approximates the $\ell^0$ measure and is easily
calculated. Under this measure, large coefficients are considered more
important than small coefficients unlike the $\ell^0$ measure. In most
settings, the $\ell^1$ solution can be used to find the support of the
$\ell^0$ solution \cite{ balan05equivalence}. The $\ell^1$ measure is
used in many optimization problems, as linear programming offers a
fast, computationally efficient solution \cite{candes05decoding,
donoho08fast}.

In \cite{karvanen03measuring} several alternative measures of sparsity
are noted which approximate the $\ell^0$ measure but emphasize
different properties. $\tanh_{a,b}$ is sometimes used in place of
$\ell^{p}$, $0<p<1$, as it is limited to the range $(0, 1)$ and
better models $\ell^0$ and $\ell^0_{\epsilon}$ in this respect. A
representation is more sparse if it has one large component, rather
than dividing up the large component into two smaller
ones. $\tanh_{a,b}$ and $\ell^p$ preserve this.  In
\cite{rickard04gini} it is shown that the $\log$ measure enforces
sparsity outside some range, but for distributions with low energy
coefficients the opposite is achieved by effectively spreading the
energy of the small components. $\kappa_4$ is the kurtosis which
measures the peakedness of a distribution
\cite{olshausen04sparse}. $u_{\theta}$ measures the smallest range
which contains a certain percentage of the data. This is achieved by
sorting the data and determining the minimum difference between the
largest and smallest sample in a range containing the specified
percentage ($\theta$) of data points as a fraction of the total range
of the data. The reason that a continuous parameter $\theta$ is used in the model is to maintain compatibility with pre-existing literature.

For measuring `diversity', \cite{rao99affine,kreutz-delgado98measures}
use some different measures. Three of these are entropy measures: the
Shannon entropy diversity measure $H_S$, a modified version of the
Shannon entropy diversity measure $H_{S}\prime $ and the Gaussian
entropy diversity measure $H_G$. They also extend the $\ell^p$ measure
to negative exponents, that is, $-1<p<0$. We call this measure
$\ell_{-}^p $ to avoid confusion.

Some of the measures can be normalized to satisfy more of the
constraints, although in general for the measures, forcing
satisfaction of one constraint means breaking another. The exception
to this is the Hoyer measure \cite{hoyer04nonnegative} which is a
normalized version of the $\frac{\ell^2}{\ell^1}$ measure as is
obvious from its definition,
$(\sqrt{N}-\frac{\ell^1}{\ell^2})(\sqrt{N}-1)^{-1}$.
\begin{figure*}[htpb]
\begin{center}
\includegraphics[width=15cm]{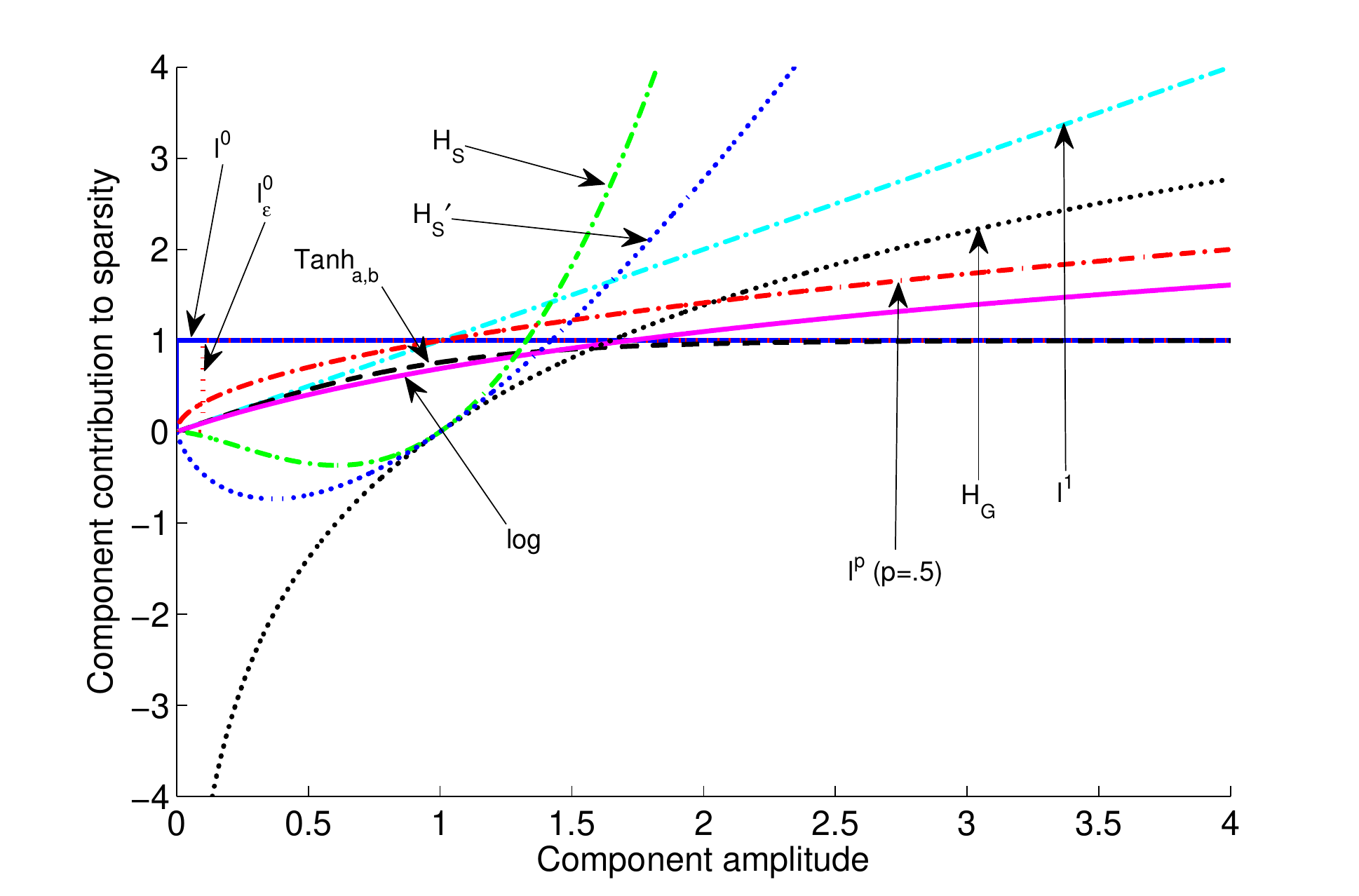}
\caption{Component contribution to sparsity measure vs component amplitude.}\label{fig:comp}
\end{center}
\end{figure*}
In Fig.~\ref{fig:comp} we can get an insight into how component
magnitude affect certain measures. In general, the smaller the
magnitude the less it impinges on the sparsity of the measure. We can
see how many of the measures approximate the $\ell^0$ measure but as
they are not flat like the $\ell^0$ measure, they have a gradient that
can be used in optimization problems. The $\ell^0$, $\ell^0_\epsilon$,
$\tanh$, $\log$, $\ell^p (0<p<1)$, $\ell^1$ measures all prefer
components to be zero or near zero. Oddly, the Shannon entropy based
measures $H_S$ and $H_S\prime$ prefer components to be at a non-zero
value less than 1.

\subsection{The Gini Index}
\label{sec:gini} Having perused the measures thus far, some
desirable aspects of a sparsity measure emerge. Like
 $\frac{\ell^2}{\ell^1}$ and Hoyer, a measure should be some kind of
 weighted sum of the coefficients. This means that unlike $\ell^0$
 when a coefficient changes slightly we have a weighted effect on the
 corresponding change in the value of the sparsity measure based on
 how `important' that particular coefficient is to the overall
 sparsity. Large coefficients should have a smaller weight than the
 small coefficients so that they do not overwhelm them to the point
 that smaller coefficients have a negligible (or no) effect on the
 measure of sparsity. If even one of the smaller coefficients is
 changed, that change should be reflected by a change in the value of
 the sparsity measure. A weighted sum achieves this.  In other words,
 we have a gradient which we can use in optimization problems. Another
 important aspect of a sparsity measure is normalization. A set of
 coefficients should not be rated more or less sparse simply because
 it has more coefficients than another set, nor should it be deemed
 more or less sparse simply due to having louder or quieter
 coefficients. In short, there should be two forms of
 normalization. Firstly, the measure of sparsity should be dependent
 on the relative values of coefficients as a fraction of the total
 value. Secondly, the measure of sparsity should be independent of the
 number of coefficients so that sets of different size can be
 compared. Lastly, it would be useful if the measure was 0 for the
 least sparse case and 1 for the most sparse case. All these qualities
 are embodied by the Gini Index, which we now define.

Given a vector, $\vc=\bmat{cccc} c_1& c_2&c_3&\cdots \emat$,
we order from smallest to largest , $c_{(1)}\leq
c_{(2)}\leq\cdots\leq c_{(N)}$ where $(1),(2),\ldots,(N)$ are the
new indices after the sorting operation. The Gini Index is given by
\begin{equation}
S(\vc) = 1 - 2\sum_{k=1}^N \frac{c_{(k)}}{\|\vc\|_1}\left(\frac{N-k+\frac{1}{2}}{N}  \right).
\end{equation}

The Gini Index also has an interesting graphical interpretation which
we see in Fig.~\ref{fig:lor}. If percentage of coefficients versus
percentage of total coefficient value is plotted for the sorted
coefficients we can define the Gini Index as twice the area between
this line and the 45$\degree$ line.  The 45$\degree$ line represents
the least sparse distribution, that with all the coefficients being
equal.

If we have a distribution from which we draw coefficients and measure
the sparsity of the coefficients which we have drawn, as we draw more
and more coefficients we would expect our measure of sparsity to
converge.  The Gini Index meets these expectations.  The Gini Index of
a distribution with probability density function $f(x)$ (which
satisfies $f(x)=0,x<0$) and cumulative distribution function $F(x)$ is
given by
\begin{equation}
  \nonumber G = 1-2\int_{0}^{1} \frac{\int_{0}^{x}tf(t) dt}{\int_{0}^{\infty}tf(t) dt}dF(x).
\end{equation}

As a side note, the Gini Index was originally proposed in economics as
measure of the inequality of wealth
\cite{lorenz05methods,gini21measurement, dalton20measurement,
arnold86majorization} and is still studied in relation to wealth
distribution as well as other
areas. \cite{shalit05mean,aaberge01axiomatic,aaberge08erratum,milanovic97simple}
`Inequality in wealth' in signal processing language is `efficiency of
representation' or `sparsity'. The utility of the Gini Index as a
measure of sparsity has been demonstrated in
\cite{rickard04gini,rickard06sparse,hurley05parameterized,hurley07maximizing}.

\begin{figure}[htpb]
\begin{center}
\begin{tabular}{c}
\includegraphics[width=5cm]{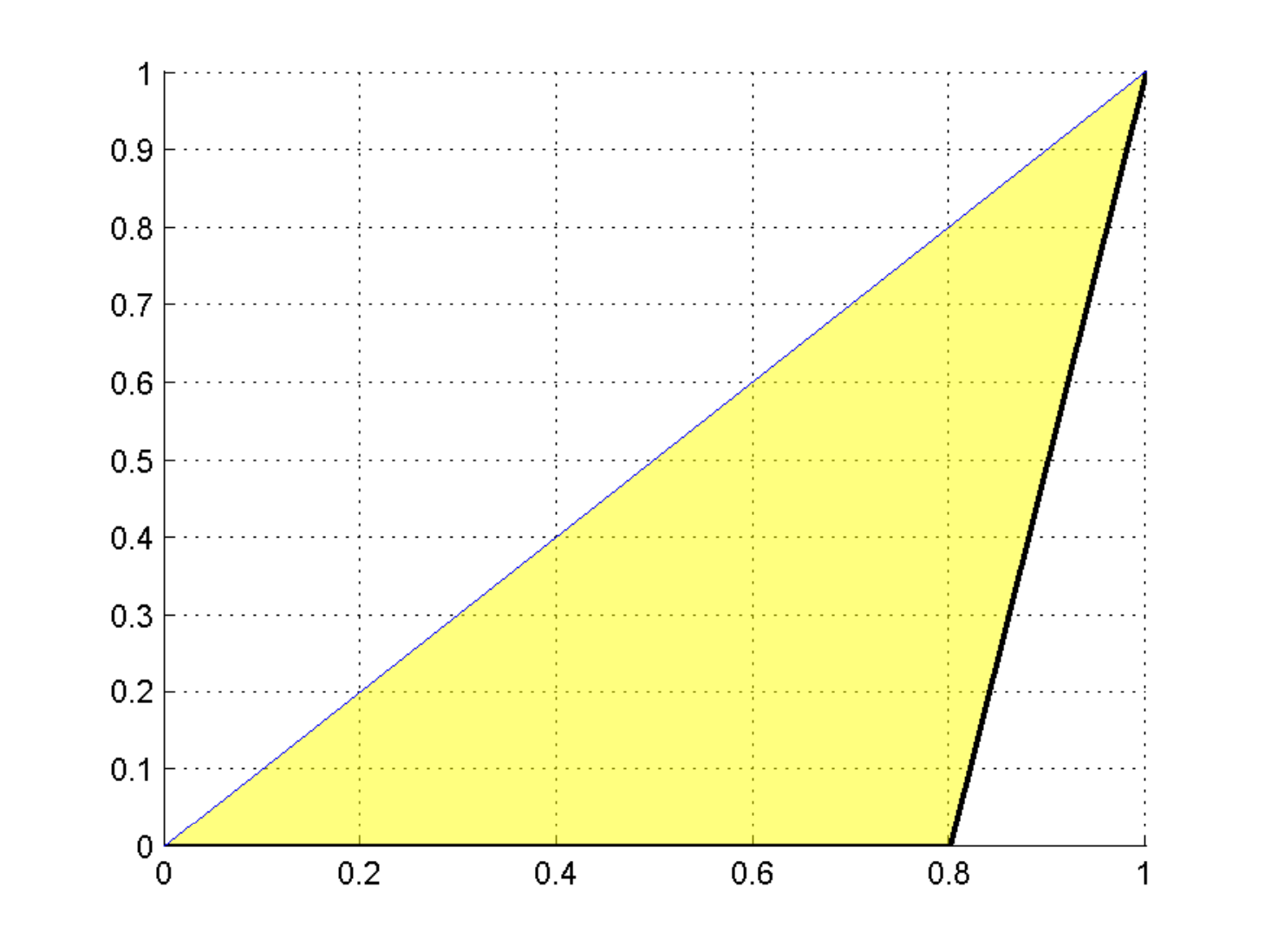}\\
\includegraphics[width=5cm]{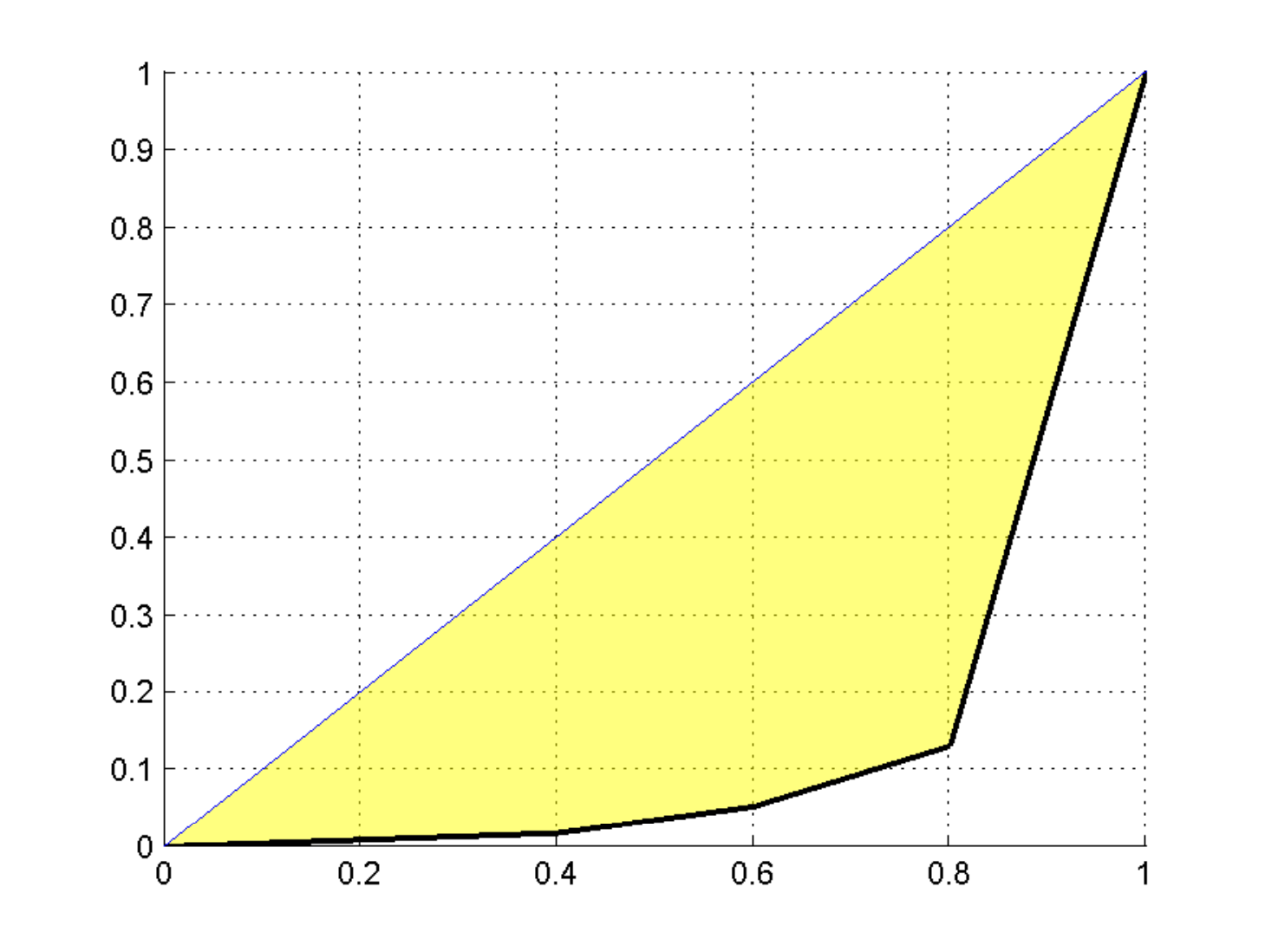}
\end{tabular}\caption{Percentage of coefficients versus percentage of total coefficient value is plotted for the sorted coefficients for  [0 0 0 0 1] (top) and [1 1 2 3 10] (bottom). The Gini Index is twice the shaded area.}\label{fig:lor}
\end{center}
\end{figure}

\section{Comparison of Sparsity Measures}
\label{sec:comp}
In this section we present the main result of the paper, the
comparison of the measures using the criteria. Many of the measures
fail for simple test cases which prove non-compliance. For example,
$[0,1,3,5]$ is more sparse than $[0,2,3,4]$ because a Robin Hood
operation maps one sequence to the other. Six of the measures do not
correctly handle this case. Others fail on similar examples. Seven of
the measures, however, satisfy \Da. An example for each sparse
criterion is given in Table~\ref{table:typcex} along with the desired
outcome when the sparsity of the examples are measured with sparsity
measure $S(\cdot)$.
\begin{table*}[ht]
\caption{Most common counter-example for a given property with measure
of sparsity and desired outcome with sparsity measure $S(\cdot)$.}
\begin{center}
\rowcolors[]{2}{blue!1}{blue!5}
\begin{tabular}{crclrcl}
Property&\multicolumn{3}{c}{Most common counter-example} &\multicolumn{3}{c}{Desired outcome}\\
\hline
\Da&$[0,1,3,5]$&vs&$[0,2,3,4]$&$S([0,1,3,5])$& $>$& $S([0,2,3,4])$ \\
\Db&$[0,1,3,5] $&vs&$[0,2,6,10]$&$S([0,1,3,5])$& $=$& $ S([0,2,6,10])$ \\
\Dc&$[1,3,5]$&vs&$[1.5,3.5,5.5]$&$S([1,3,5])$& $<$& $S([1.5,3.5,5.5])$ \\
\Dd&$[0,1,3,5]$&vs&$[0,0,1,1,3,5]$&$S([0,1,3,5])$& $=$& $S([0,0,1,1,3,5])$ \\
\De&$[0,1,3,5]$&vs&$[0,1,3,20]$&$S([0,1,3,5])$& $<$& $S([0,1,3,20])$ \\
\Df&$[0,1,3,5]$&vs&$[0,0,0,1,3,5]$&$S([0,1,3,5])$& $<$& $S([0,0,0,1,3,5])$
\end{tabular}

\label{table:typcex}
\end{center}
\end{table*}
Table~\ref{table:checks} details which of the six sparse criteria hold
for each of the fifteen measures. The information is based on proofs
and counter-examples which are contained in their entirety in
Appendices~\ref{sec:cex}~and~\ref{sec:proofs}.  There are essentially
two types of proof, Type~A and Type~B. Type~A is the standard form of
proof which uses inequalities, an example of which is the following:
\begin{thm}
  $\frac{\ell^2}{\ell^1}$   satisfies \[S(\bmat{cccccc}c_1& \cdots& c_i-\alpha&
\cdots& c_j+\alpha& \cdots\emat ) < S(\vc),\] for all $\alpha, c_i, c_j $
such that $c_i>c_j$ and $0<\alpha < \frac{c_i-c_j}{2}$.
\end{thm}
\begin{proof}
\label{proof:A}
As $\frac{\ell^2}{\ell^1} = \frac{\sqrt{\sum_j c_j^2}}{\sum_j
c_j}$ we can restate the above as
\[
\frac{\sqrt{\sum_{k \neq i,j}c_k^2 +(c_i -\alpha)^2  + (c_j +\alpha)^2 }}{\sum_k c_k+\alpha -\alpha} < \frac{\sqrt{\sum_k c_k^2} }{\sum_k c_k}.
\]
This simplifies to
%% \begin{eqnarray*}
%% && \sum_{k \neq i,j} c_k^2 +(c_i -\alpha)^2  + (c_j +\alpha)^2 < \sum_k c_k^2\\
%% && (c_i -\alpha)^2  + (c_j +\alpha)^2 < c_i^2 + c_j^2\\
%% \end{eqnarray*}
\begin{equation}
  \nonumber (c_i -\alpha)^2  + (c_j +\alpha)^2 < c_i^2 + c_j^2.
\end{equation}
Expand this to get
\begin{eqnarray*}
c_i^2 -2c_i \alpha + \alpha^2  + c_j^2 +2c_j \alpha + \alpha^2 &<& c_i^2 + c_j^2\\
c_j - c_i + \alpha &<&0,
\end{eqnarray*}
which we know is true as $0<\alpha < \frac{c_i-c_j}{2}$.
\end{proof}

A type~B proof on the other hand uses
derivatives, for example:
\begin{thm}
 $-{\ell^p}$ satisfies  \[S(\bmat{cccccc}c_1& \cdots& c_i-\alpha&
\cdots& c_j+\alpha& \cdots\emat ) < S(\vc),\] for all $\alpha, c_i, c_j $
such that $c_i>c_j$ and $0<\alpha < \frac{c_i-c_j}{2}$.
\end{thm}
\begin{proof}
\label{proof:B}
\[
-{\ell^p} = -{\left(\sum_k c_k^p\right)^{1/p},\,\,\,0<p<1}.
\]
We wish to show that the following  holds true  for all $\alpha , c_i, c_j $ such that $c_i>c_j$ and $0<\alpha < \frac{c_i-c_j}{2}$
\[
\frac{\partial}{\partial \alpha}\left[-{\left(\sum_{n \neq i,j} c_n^p  +(c_i-\alpha)^p + (c_j+\alpha)^p \right)^{1/p}}\right] <0.
\]
Expand this to get
{\small
\begin{eqnarray*}
&   -\frac{1}{p} \left(\sum_{k\neq i,j} c_k^p +(c_i-\alpha)^{p}+(c_j+\alpha)^{p}\right)^{\frac{1}{p}-1}
\left(-p(c_i-\alpha)^{p-1}\right.&\\
&\left. +p(c_j+\alpha)^{p-1}\right)<0.&
\end{eqnarray*}}
Which holds true if
\begin{equation}
  \nonumber (c_j+\alpha)^{p-1}-(c_i-\alpha)^{p-1}>0.
\end{equation}
As $p-1<0$ we can rewrite the above as
\begin{eqnarray*}
\frac{1}{ (c_j+\alpha)^{1-p}}-\frac{1}{(c_i-\alpha)^{1-p}}&>&0\\
\frac{1}{ (c_j+\alpha)}&>&\frac{1}{(c_i-\alpha)}\\
c_i-\alpha&>&c_j+\alpha\\
\frac{c_i-c_j}{2}&>&\alpha,
\end{eqnarray*}
which is necessarily true as it is one of the constraints upon $\alpha$.
\end{proof}
From Table~\ref{table:checks} we can see that \Dc~({\em
 Rising Tide}) is satisfied by most measures. This shows that relative
 size of coefficients is of the utmost importance when desiring
 sparsity. As previously mentioned, most measures do not satisfy
 \Dd~({\em Cloning}). Each of the other criteria is satisfied
 by a varying number of the fifteen measures of sparsity. This
 demonstrates the variety of attributes to which measures of sparsity
 attach importance. $\kappa_4$ and the Hoyer measure satisfy most of the
 criteria. The Gini Index alone satisfies all six criteria.
%___________________________________________________________________________%
{\large
\begin{table}[ht]
\caption{Comparison of different sparsity measures using criteria defined in  Sec.~\ref{sec:crit}}
\begin{center}
%%\rowcolors[]{2}{blue!1}{blue!5}
\begin{tabular}{|c||c|c|c|c|c|c|}
\hline
Measure &  {\bf \Da}&
{\bf \Db}&
{\bf \Dc}&
{\bf \Dd}&
{\bf \De}&
{\bf \Df}\\
\hline
\hline
%measure                  % D1         %D2          %D3          %D4          %P1          %P2
$\ell^0$
                      &            & \checkmark     &              &            &            & \checkmark\\
\hline
$\ell^0_\epsilon$
                      &            &      &              &            &            & \checkmark\\
\hline
$-{\ell^1}$
                      &            &                &  \checkmark &            &            &           \\
\hline
$-{\ell^p}$
                     & \checkmark  &                &  \checkmark &            &            &           \\
\hline
$\frac{\ell^2}{\ell^1}$
                     & \checkmark  & \checkmark     &             &            &\checkmark  &           \\
\hline
$-{\tanh_{a,b}}$
                      &\checkmark  &                &  \checkmark            &            &            &           \\
\hline
  $-\log$
                      &            &                &  \checkmark &            &            &           \\
\hline
  $\kappa_4$
                      &            & \checkmark     &  \checkmark &            &  \checkmark &           \\
\hline
  $u_\theta$
                      &            &     \checkmark           &             &\checkmark  &  \checkmark   &\\
\hline
  $-\ell^p_{-}$
                      &  &                &   &            &  \checkmark     &           \\
\hline
  $H_G$
                      & \checkmark &                &  \checkmark &            &            &        \\
\hline
 $H_S$
                      &            &                &   &            &            &           \\
\hline
$H_S^\prime$
                      &            &                &   &            &            &    \\
\hline
 Hoyer
                      &\checkmark  & \checkmark     & \checkmark  &            & \checkmark & \checkmark \\
\hline
Gini
                      & \checkmark & \checkmark     &  \checkmark & \checkmark & \checkmark & \checkmark\\
\hline

\end{tabular}
\label{table:checks}
\end{center}
\end{table}}

%___________________________________________________________________________%
\subsection{Numerical Sparse Analysis}
\label{sec:ana}
In this section we present the results of using the fifteen sparse
measures to measure the sparsity of data drawn from a set of
parameterized distributions. We select data sets and distributions for
which we can change the `sparsity' by altering a parameter. By
applying the fifteen measures to data drawn from these distributions
as a function of the parameter, we can visualize the criteria. The
examples are based on the premise that all coefficients being equal is
the least sparse scenario and all coefficients being zero except one
is the most sparse scenario.

In the first experiment we draw a variable number of coefficients from
a probability distribution and measure their sparsity. We expect sets
of coefficients from the same distribution to have a similar
sparsity. As we increase the number of coefficients we expect the
measure of sparsity to converge. In this experiment we examine the
sparsity of sets of coefficients from a Poisson distribution
(Fig.~\ref{fig:poissplot})
\begin{figure*}[htpb]
\begin{center}
\includegraphics[width=15cm, height = 9cm]{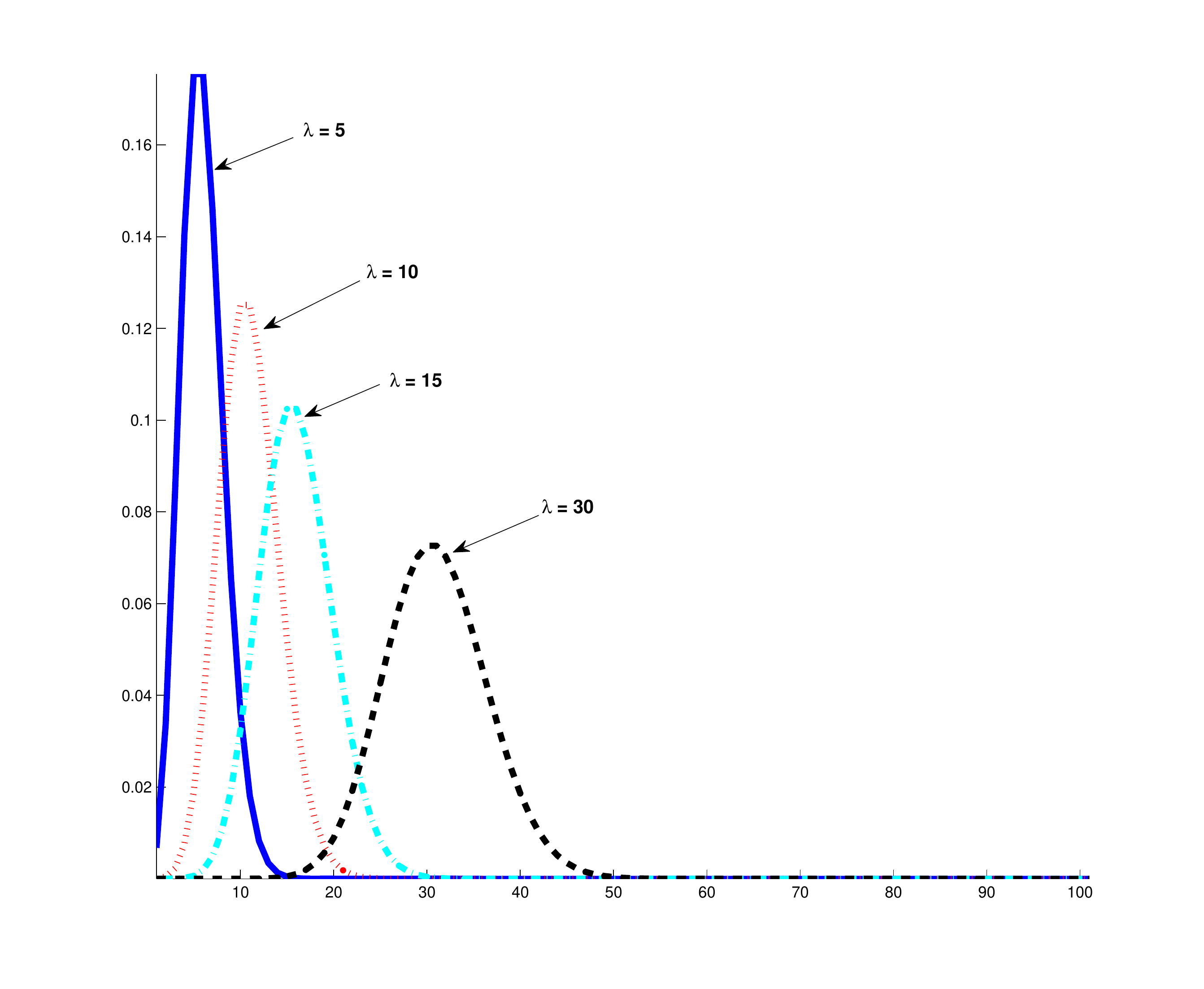}
\caption{Sample Poisson distribution probability density functions for $\lambda = 5,10,15,30$. We expect the distributions with a `narrower' peak (small $\lambda$) to have a higher sparsity than those with a `wider' peak (large $\lambda$)}
\label{fig:poissplot}
\end{center}
\end{figure*}
with parameter $\lambda = 5$ as a function of set size.  From the
normalized version of the sparsity plot in Fig.~\ref{fig:d4alt} we can
see that three measures converge. They are $\kappa_4$, the Hoyer measure
and the Gini Index. As this is similar in nature to \Dd~we expect the
Gini Index to converge. The convergence of Hoyer measure is
unsurprising as this measure almost satisfies \Dd~ especially for
large $N$.  The results are also normalized for clearer visualization
in that they are modified so that the sparsity falls between 0 and 1.
\begin{figure*}[htpb]
\begin{center}
%% \begin{tabular}{cc}
%% \includegraphics[width=4cm]{D4_alt}&\includegraphics[width=2cm,height=4.5cm]{legend}\\
%% \multicolumn{2}{l}
{\includegraphics[width=15cm]{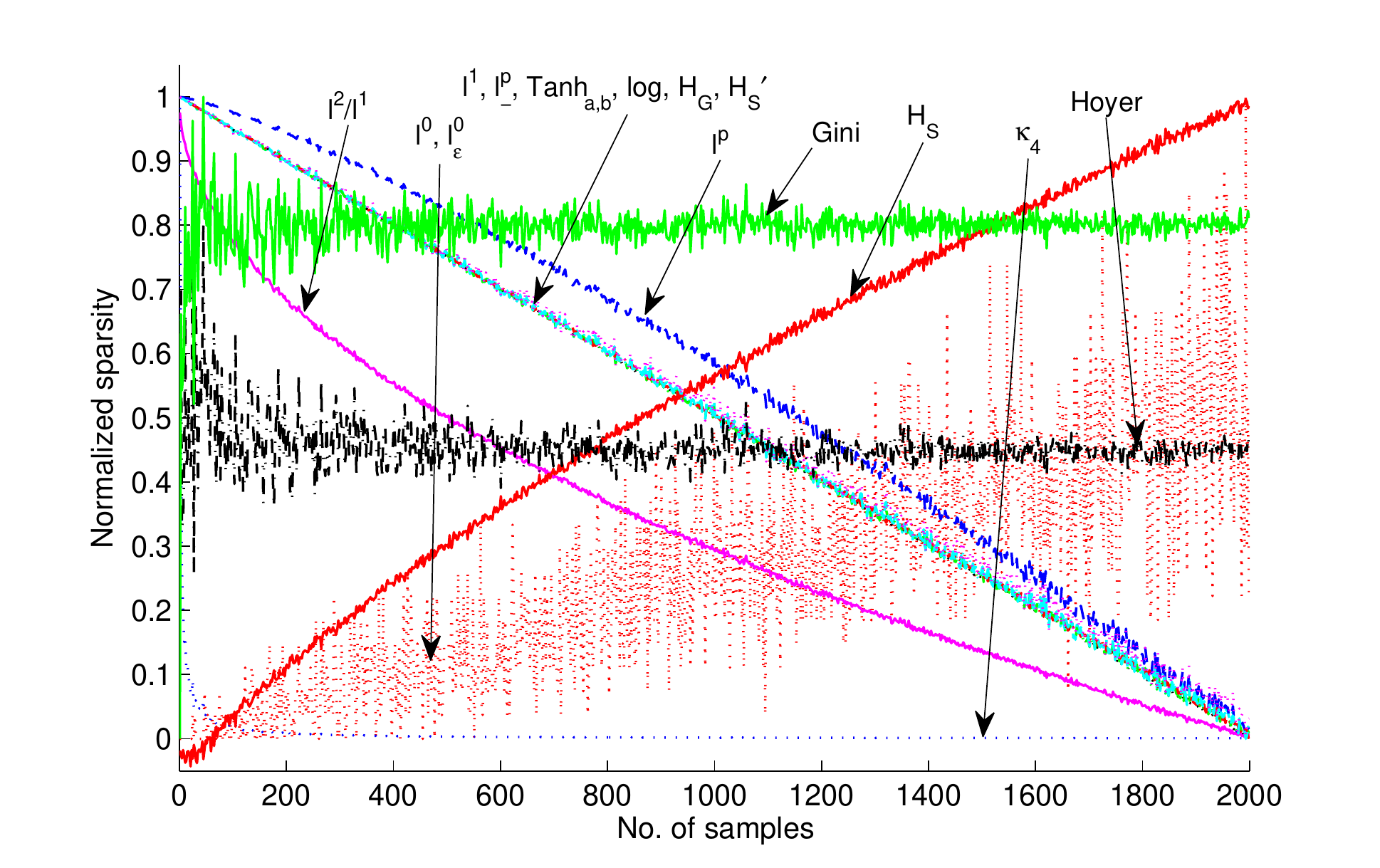}}\\
%(a) Sparsity vs no. of coefficients & (b) Sparsity vs no. of coefficients (normalized )\\
%%\end{tabular}
\caption{\small{Sparsity of sets of coefficients drawn from a Poisson
 distribution ($\lambda=5$) {\it vs} the length of the vector of
 coefficients. The erratically ascending measures are $\ell^0$ and
 $\ell^0_{\epsilon}$. The measures $\ell^1$, $\log$, $\tanh$, $H_G$,
 $H_S^\prime$ and $\ell^p_{-}$ are grouped in an almost-straight
 decreasing line. The measures are scaled to be between 0 and 1.}
 }\label{fig:d4alt}
\end{center}
\end{figure*}

In the second experiment we take coefficients from a Bernoulli
distribution where coefficients are either $0$ with probability $p$ or
$1$ with probability $1-p$. For this experiment the set size remains
constant and the probability $p$ varies from $0$ to $1$.  With a low
$p$ most coefficients will be 1 and very few zero. The energy
distribution of such a set is not sparse and accordingly has a low
value (see Fig~\ref{fig:bern}). As $p$ increases so should the
sparsity measure. We can see this is the case in some form for all of
the measures except $H_S \prime$. We note that  $\kappa_4$
does not rise steadily with increasing $p$ but rises dramatically as
the set approaches its sparsest. This is of some concern if optimizing
sparsity using $\kappa_4$ as there is not much indication that the
distribution is getting more sparse until its already quite
sparse.
\begin{figure*}[htpb]
\begin{center}
\includegraphics[width=14cm]{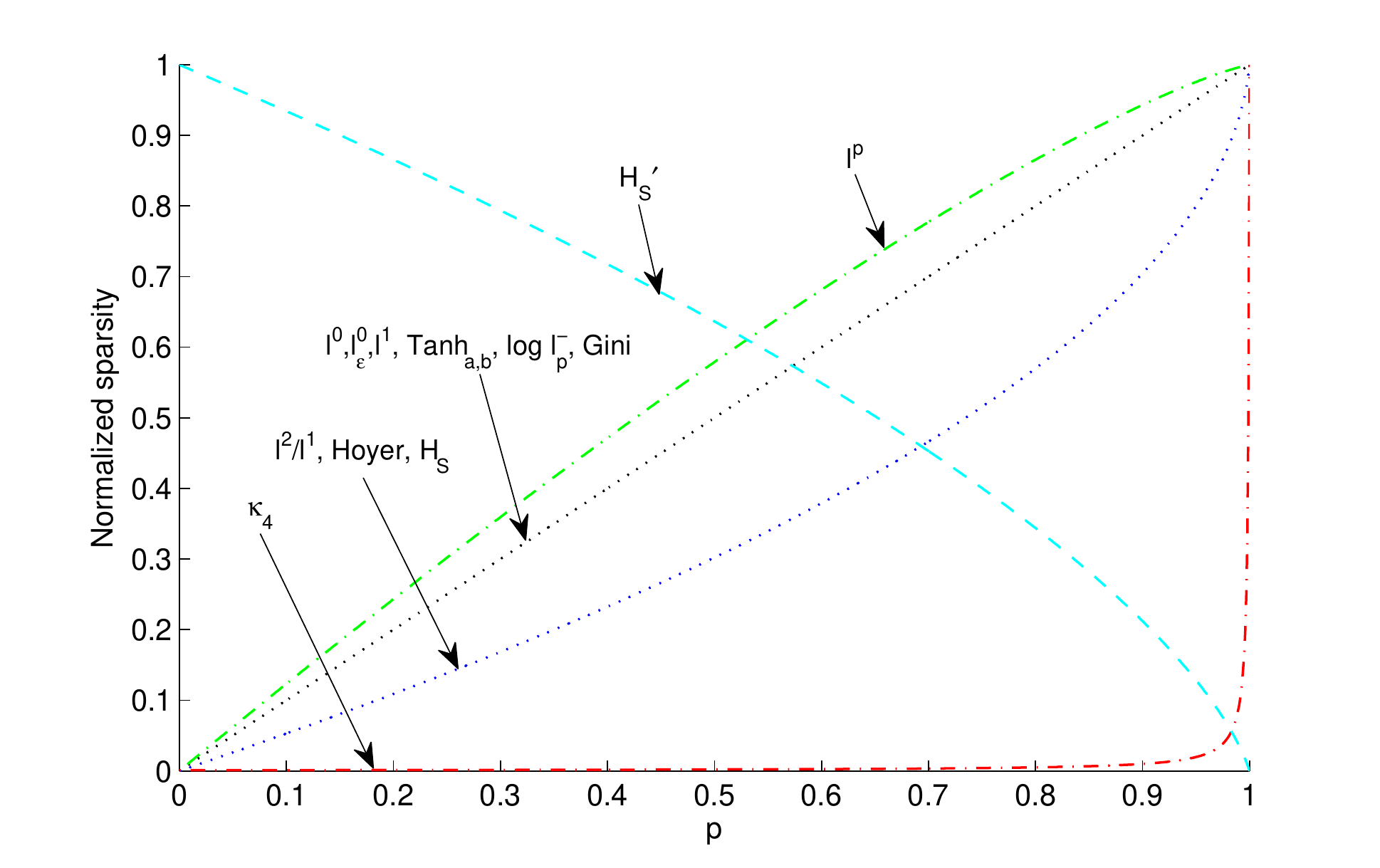}
\caption{Sparsity {\it vs} $p$ for a Bernoulli distribution with
coefficients being $0$ with probability $p$ and $1$ otherwise. The
measures are scaled to fit between a sparsity range of 0 to 1.}
\label{fig:bern}
\end{center}
\end{figure*}

\section{Conclusions}
\label{sec:concl} In this paper we have presented six intuitive
attributes of a sparsity measure. Having defined these attributes
mathematically, we then compared commonly-used measures of sparsity.
The goal of this paper is to provide motivation for selecting a
particular measure of sparsity. Each measure emphasizes different
combinations of attributes and this should be addressed when
selecting a sparsity measure for an application. We can see from the
main contribution of this paper, Table~\ref{table:checks} and the
associated proofs in Appendices~\ref{sec:cex} and \ref{sec:proofs},
that the only measure to satisfy all six criteria is the Gini Index.
This aligns well with \cite{rickard06sparse} in which it is shown
that the Gini Index is an indicator for when sources are separable,
a property which itself relies on sparsity. The Hoyer measure
\cite{hoyer04nonnegative} comes a close second, failing only
\Dd~(invariance under cloning), which is, admittedly an arguable
criterion for certain applications. For applications in which the
number of coefficients is fixed both the Gini Index and the Hoyer
measure satisfy all criteria.
%The kurtosis measure ($\kappa_4$)
%fails both \Dd~(invariance under cloning) and \Df~(Babies increase
%sparsity). This behavior could be desirable under certain conditions
%for example in a system where signals were zero-padded.

We have also presented two graphical examples of the performance of
the measures when quantifying the sparsity of a distribution with
sparsity controlled. Again, both the Gini Index and the Hoyer
measure outperform the other measures, illustrating their utility.

Sparsity is used in many applications but with few exceptions it is
not studied as a concept in itself. We hope that this work will not
just encourage the use of the Gini Index but encourage users of
sparsity to consider in more depth the concept of
sparsity.

\appendix
%\section{Appendix}
%\label{sec:appa}
We use these measures to calculate a number which describes the
sparsity of a set of coefficients $\vc = \bmat{cccc}c_1& c_2 & \cdots
&c_N\emat$.

Note - ignore the trivial cases, for example,   \Db~  with  $\alpha=1$.
\begin{itemize}
\item[\Da] {\em Robin~Hood}:\\{\small $S(\bmat{cccccc}c_1& \cdots& c_i-\alpha&
\ldots& c_j+\alpha& \ldots\emat )$ $<S(\vc)$} for all $\alpha, c_i, c_j $
such that $c_i>c_j$ and $0<\alpha < \frac{c_i-c_j}{2}$.
\item[\Db] {\em Scaling}:\\ $S(\alpha\vc) = S(\vc)$, $\forall \alpha \in \mathbb{R},~\alpha > 0$.
\item[\Dc] {\em Rising Tide}:\\$S(\alpha + \vc)<S(\vc)$, $\alpha \in
\mathbb{R},~\alpha > 0$ (We exclude the case $c_1 = c_2 = c_3 = \cdots
= c_i = \cdots \forall i$ as this is equivalent to scaling.).
\item[\Dd] {\em Cloning}:\\$ S(\vc) = S(\vc\|\vc) = S(\vc\|\vc\|\vc) = S(\vc\|\vc\|\cdots\|\vc)$.
\item[\De] {\em Bill Gates}:\\  $\forall i \exists \beta = \beta_i>0$, such that $\forall \alpha >0:$ {\small \[S(\bmat{cccc}c_1& \ldots& c_i+\beta+\alpha&\ldots\emat)>S(\bmat{cccc}c_1&\ldots&c_i+\beta& \ldots\emat).\]}
\item[\Df] {\em Babies}:\\$S(\vc||0) > S(\vc)$.
\end{itemize}
\begin{footnotesize}
\begin{sidewaystable*}
%\begin{table}[ht]
\caption{Guide to Counter-examples and Proofs each followed by
reference number. A \checkmark indicates compliance of the measure
with the relevant criterion. `obv' means that the proof is obvious
and as such is not included.}
\begin{center}
\rowcolors[]{2}{blue!1}{blue!5}
\begin{tabular}{ccccccc}

Measure &  {\bf \Da}&
{\bf \Db}&
{\bf \Dc}&
{\bf \Dd}&
{\bf \De}&
{\bf \Df}\\
\hline
%\hline
%measure                  % D1         %D2          %D3          %D4          %P1          %P2
$\ell^0$
        & C.Ex~\ref{CE:1}              &  \checkmark obv             & C.Ex~\ref{CE:3}     &     C.Ex~\ref{CE:4}       &        C.Ex~\ref{CE:5}      & \checkmark obv\\
%\hline
$\ell^0_\epsilon$
                      &  C.Ex~\ref{CE:1}     &  C.Ex~\ref{CE:2}     &  obv&  C.Ex~\ref{CE:4}         &  C.Ex~\ref{CE:5}   & \checkmark obv\\
%\hline
$-{\ell^1}$
             & C.Ex~\ref{CE:1}    &    C.Ex~\ref{CE:2}  &   \checkmark  obv        & C.Ex~\ref{CE:4}         &   C.Ex~\ref{CE:5}        &    C.Ex~\ref{CE:6} \\
%\hline
$-{\ell^p}$
            & \checkmark Proof~\ref{sec:lp1}     &    C.Ex~\ref{CE:2}     &   \checkmark   Proof~\ref{sec:lp3}            &   C.Ex~\ref{CE:4}    &      C.Ex~\ref{CE:5}       &  C.Ex~\ref{CE:6}                     \\
%\hline
$\frac{\ell^2}{\ell^1}$
                     & \checkmark  Proof~\ref{sec:l1l21} & \checkmark obv     & Proof~\ref{sec:l1l22}      &  C.Ex~\ref{CE:4}          &\checkmark Proof~\ref{sec:l1l23} &    C.Ex~\ref{CE:6}       \\
%\hline
$-{\tanh_{a,b}}$
                      &\checkmark  Proof~\ref{sec:tanh2}&    C.Ex~\ref{CE:3}       &   \checkmark  Proof~\ref{sec:tanh3}          &    C.Ex~\ref{CE:4}        &       C.Ex~\ref{CE:5}         &    C.Ex~\ref{CE:6}  \\
%\hline
  $-\log$
                      &   C.Ex~\ref{CE:1a} $(^{\ast})$         &    C.Ex~\ref{CE:2}             &  \checkmark Proof~\ref{sec:log3}&    C.Ex~\ref{CE:4}    &     C.Ex~\ref{CE:5}       &   C.Ex~\ref{CE:6}     \\
%\hline
  $\kappa_4$
                      &  C.Ex~\ref{CE:1a} $(^{\ast})$     & \checkmark Proof~\ref{sec:kur2}    &  \checkmark Proof~\ref{sec:kur3} & C.Ex~\ref{CE:4}           &  \checkmark Proof~\ref{sec:kur5}&      C.Ex~\ref{CE:6}       \\
%\hline
  $u_\theta$
                      & Proof~\ref{sec:u1}          &     \checkmark   obv   &     Proof~\ref{sec:u3}        &\checkmark  Proof~\ref{sec:u4}  &  \checkmark Proof~\ref{sec:u5}         & Proof~\ref{sec:u6}  \\
%\hline\
  $-\ell^p_{-}$
                      &C.Ex~\ref{CE:1}  &     C.Ex~\ref{CE:2}           &  C.Ex~\ref{CE:3a}$(^{\ast})$&    C.Ex~\ref{CE:4}        &     \checkmark Proof~\ref{sec:lp_5}     &      C.Ex~\ref{CE:6}      \\
%\hline
  $H_G$
                      & \checkmark Proof~\ref{sec:hg1} &         C.Ex~\ref{CE:2}        &  \checkmark obv &   C.Ex~\ref{CE:4}         &    C.Ex~\ref{CE:5}         &    C.Ex~\ref{CE:6}   \\
%\hline
 $H_S$
                      &    C.Ex~\ref{CE:1}         &        C.Ex~\ref{CE:2}         &  C.Ex~\ref{CE:3a}$(^{\ast})$ &    C.Ex~\ref{CE:4}        &    C.Ex~\ref{CE:5}  &  C.Ex~\ref{CE:6}   \\
%\hline
 $H_S^\prime$
                      &  C.Ex~\ref{CE:1}           &        C.Ex~\ref{CE:2}         &  C.Ex~\ref{CE:3a}$(^{\ast})$&  C.Ex~\ref{CE:4}          &  C.Ex~\ref{CE:5}       &   C.Ex~\ref{CE:6}    \\
%\hline
 Hoyer
                      &\checkmark Proof~\ref{sec:hoy1} & \checkmark   obv  & \checkmark  Proof~\ref{sec:hoy3} & C.Ex~\ref{CE:4}            & \checkmark Proof~\ref{sec:hoy5}& \checkmark obv \\
%\hline
Gini
                      & \checkmark Proof~\ref{sec:gini1} & \checkmark  Proof~\ref{sec:gini2}   &  \checkmark Proof~\ref{sec:gini3} & \checkmark Proof~\ref{sec:gini4} & \checkmark Proof~\ref{sec:gini5}  & \checkmark Proof~\ref{sec:gini6} \\
%\hline
\end{tabular}
\label{table:counterexs}
\end{center}
%\end{table}
\end{sidewaystable*}
\end{footnotesize}

\subsection{Counter-Examples}
\label{sec:cex}
The most parsimonious method of showing non-compliance with the sparse
criteria is through the following simple counter-examples.  As an
sample we take the $-\ell^1$ measure and \Da. \Da~ states that the
$\ell^1$ measure of $[0,1,3,5]$ should be greater than the $\ell^1$
measure of $[0,2,3,4]$. Using counter example \label{CE:1} we see that
\begin{eqnarray*}
 S([0,1,3,5]) &=& -9\\
S([0,2,3,4]) &=&-9.
\end{eqnarray*}
 As the Robin Hood operation had no effect on the sparsity of the
vectors as measured by the $\ell^0$ measure the measure does not
satisfy \Da.  In the case of $-\ell^p_{-}$ the zeros in the
counter-examples are omitted.
\begin{counterex}%{\em Counter Example 1}
\label{CE:1}
\[
[0,1,3,5] ~~\mbox{vs}~~ [0,2,3,4]
\]
\end{counterex}
\setcounter{niall11}{0}
\begin{counterex}[$^{\ast}$]%{\em Counter Example 1*}

\label{CE:1a}
\[
[.3, 1, 2] ~~\mbox{vs}~~  [.31, .99,  2]
\]
\end{counterex}
\begin{counterex}%{\em Counter Example 2}
\label{CE:2}
\[
[0,1,3,5] ~~\mbox{vs}~~ [0,2,6,10]
\]
\end{counterex}
% especially for l0 and l0_eps

\begin{counterex}%{\em Counter Example 3}
\label{CE:3}
\[
[1,3,5] ~~\mbox{vs}~~ [1.5,3.5,5.5]
\]
\end{counterex}
% especially for lp-
\setcounter{niall11}{2}
\begin{counterex}[$^{\ast}$]%{\em Counter Example 3*}
\label{CE:3a}
\[
[.1,.3,.5] ~~\mbox{vs}~~ [.15, .35,.55]
\]
\end{counterex}
\begin{counterex}%{\em Counter Example 4}
\label{CE:4}
\[
[0,1,3,5]~~\mbox{vs}~~  [0,0,1,1,3,5]
\]
\end{counterex}
\begin{counterex}%{\em Counter Example 5}
\label{CE:5}
\[
[0,1,3,5]~~\mbox{vs}~~ [0,1,3,20]
\]
\end{counterex}
\begin{counterex}%{\em Counter Example 6}
\label{CE:6}
\[
[0,1,3,5] ~~\mbox{vs}~~ [0,0,0,1,3,5]
\]
\end{counterex}
\subsection{Proofs}
\label{sec:proofs}
This section contains the proofs that were longer than
Table~\ref{table:counterexs} permitted. The obvious method of proving
that the measures satisfy the criteria, is to plug the formulae
for the measures into the mathematical definitions of the six
criteria. Another method used below is to differentiate the modified
sparse measure with respect to the parameter that modifies it and
observe the result. For example if we show that $\frac{\partial
S(\alpha + \vc)}{ \partial \alpha} < 0$ for $\alpha >0$ this proves
\Dc~ as any change in $\alpha$ causes the measure to drop.

%\begin{onecolumn}
\subsubsection{$-{\ell^p}$  and \Da}\hfill
\label{sec:lp1}
\begin{thm}
 $-{\ell^p}$ satisfies
\[
S(\bmat{cccccc}c_1& \cdots& c_i-\alpha& \cdots& c_j+\alpha&
\cdots\emat ) < S(\vc)
\]
 for all $\alpha, c_i, c_j $ such that $c_i>c_j$ and $0<\alpha <
\frac{c_i-c_j}{2}$.
\end{thm}
 \begin{proof}
\[
-{\ell^p} = -{\left(\sum_k c_k^p\right)^{1/p},\,\,\,0<p<1}.
\]
We wish to show that the following holds true  for all $\alpha , c_i, c_j $ such that $c_i>c_j$ and $0<\alpha < \frac{c_i-c_j}{2}$
\[
\frac{\partial}{\partial \alpha}\left[-{\left(\sum_{n \neq i,j} c_n^p  + (c_i-\alpha)^p + (c_j+\alpha)^p \right)^{1/p}}\right] <0.
\]
{\small
\begin{eqnarray*}
&   -\frac{1}{p} \left(\sum_{k\neq i,j} c_k^p +(c_i-\alpha)^{p}+(c_j+\alpha)^{p}\right)^{\frac{1}{p}-1}
\left(-p(c_i-\alpha)^{p-1}\right.&\\
&\left. +p(c_j+\alpha)^{p-1}\right)<0.&
\end{eqnarray*}}
Which holds true if
\begin{equation}
  \nonumber (c_j+\alpha)^{p-1}-(c_i-\alpha)^{p-1}>0.
\end{equation}
As $p-1<0$ we can rewrite the above as
\begin{eqnarray*}
\frac{1}{ (c_j+\alpha)^{1-p}}-\frac{1}{(c_i-\alpha)^{1-p}}&>&0\\
\frac{1}{ (c_j+\alpha)}&>&\frac{1}{(c_i-\alpha)}\\
c_i-\alpha&>&c_j+\alpha\\
\frac{c_i-c_j}{2}&>&\alpha,
\end{eqnarray*}
which is necessarily true as it is one of the constraints upon $\alpha$.
\end{proof}
\subsubsection{$-{\ell^p}$  and \Dc}\hfill
\label{sec:lp3}
\begin{thm}
$-{\ell^p}$   satisfies \[S(\alpha + \vc)<S(\vc), ~\alpha \in \mathbb{R},~\alpha > 0.\]
\end{thm}
 \begin{proof}
  \begin{eqnarray*}
 -\left(\sum_{k=1}^N (\alpha+c_k)^p\right)^{1/p}
 &<& -\left(N\alpha^p + \sum_{k=1}^N c_k^p\right)^{1/p} \\
&<&  -\left(\sum_{k=1}^N c_k^p\right)^{1/p}.  \end{eqnarray*}
\end{proof}
\subsubsection{$\frac{\ell^2}{\ell^1}$  and \Da}\hfill
\label{sec:l1l21}
\begin{thm}
$\frac{\ell^2}{\ell^1}$  satisfies \[S(\bmat{cccccc}c_1& \cdots& c_i-\alpha& \cdots& c_j+\alpha& \cdots\emat ) < S(\vc)\] for
all $\alpha, c_i, c_j $ such that $c_i>c_j$ and $0<\alpha <
\frac{c_i-c_j}{2}$.
\end{thm}
\begin{proof}
As $\frac{\ell^2}{\ell^1} = \frac{\sqrt{\sum_j c_j^2}}{\sum_j c_j}$ we can restate the above as
\[
\frac{\sqrt{\sum_{k \neq i,j}c_k^2 +(c_i -\alpha)^2  + (c_j +\alpha)^2 }}{\sum_k c_k+\alpha -\alpha} < \frac{\sqrt{\sum_k c_k ^2} }{\sum_k c_k}.
\]
This simplifies to
\begin{eqnarray*}
 \sum_{k \neq i,j} c_k ^2 +(c_i -\alpha)^2  + (c_j +\alpha)^2 &<& \sum_k c_k ^2\\
 (c_i -\alpha)^2  + (c_j +\alpha)^2 &<& c_i^2 + c_j^2.
\end{eqnarray*}
\begin{eqnarray*}
c_i^2 -2c_i \alpha + \alpha^2  + c_j^2 +2c_j \alpha + \alpha^2 &<& c_i^2 + c_j^2\\
 c_j - c_i + \alpha &<&0,
\end{eqnarray*}
which we know is true as $0<\alpha < \frac{c_i-c_j}{2}$.
\end{proof}
\subsubsection{$\frac{\ell^2}{\ell^1}$  and \Dc}\hfill
\label{sec:l1l22}
\begin{thm}
$\frac{\ell^2}{\ell^1}$  does not satisfy \[S(\alpha + \vc)<S(\vc), ~\alpha \in \mathbb{R},~\alpha > 0.\]
\end{thm}
 \begin{proof}
\[
\frac{\sqrt{\sum_j (\alpha + c_j)^2}}{\sum_j (\alpha + c_j)} = \frac{\sqrt{\sum_j (\alpha^2  +2c_j\alpha+ c_j^2)}}{\sum_j (\alpha + c_j)}.
\]
To simplify matters we make the following substitutions
\begin{eqnarray}
\nonumber \sa &=&\sum_jc_j\\
\saa &=&\sum_j c_j^2
\label{eq:subl1l2}
\end{eqnarray}
 and note that $\sa^2 > \saa$. We now have
 \beq
  \frac{\sqrt{ \saa +2\alpha \sa+N\alpha^2}}{\sa +N\alpha} &<& \frac{\sqrt{\saa}}{\sa}\\
 \sa^2(\saa +2\alpha\sa+N\alpha^2) &<& \saa(\sa^2 +2\sa N\alpha +N^2\alpha^2)\\
  \alpha &<& \frac{N}{2\sa}\left(\frac{\saa - \sa^2 }{N \sa^2 - \saa}\right),
\eeq
 which is false as $\left(\frac{\saa - \sa^2 }{N \sa^2 -
\saa}\right)<0$ which violates the condition $\alpha > 0$.
\end{proof}
\subsubsection{$\frac{\ell^2}{\ell^1}$  and \De}\hfill
\label{sec:l1l23}
\begin{thm}
$\frac{\ell^2}{\ell^1}$  satisfies $\forall i \exists \beta = \beta_i>0$,  $\forall \alpha >0:$ {\small \[S(\bmat{cccc}c_1& \ldots& c_i+\beta+\alpha&\ldots\emat) > S(\bmat{cccc}c_1&\ldots&c_i+\beta& \ldots\emat).\]}
\end{thm}
 \begin{proof}
We make the following substitutions
\beq
\nonumber \sa &=&\sum_j {c_j}\\
\saa &=&\sum_j {c_j}^2
\eeq
and wish to show that
{\small
\[
\frac{
\sqrt{\saa +\alpha^2 +\beta^2 +2(\alpha\beta+ \alpha{c_i} + \beta{c_i})}}{\sa   + \alpha+\beta}  > \frac{\sqrt{\saa +\beta^2 +2 c_i \beta}}{\sa  +\beta }.
\]}
Squaring both sides and cross-multiplying gives
%% \begin{footnotesize}
%%   \begin{equation}
%%   \nonumber (\saa +\alpha^2 +\beta^2 +2\alpha\beta+ 2\alpha{c_i} + 2\beta{c_i})(\sa +\beta )^2  >(\saa +\beta^2 +2 c_i \beta)(\sa  + \alpha+\beta)^2
%%   \end{equation}
%% \end{footnotesize}
%% which, after expanding,  can be simplified to
%% {\small
%% \begin{equation}
%%   \nonumber \alpha^2 \sa^2 + 2\alpha\beta\sa^2 + 2 \alpha c_i \sa^2 +2\sa\alpha^2 \beta + 2 \sa \beta ^2 > \saa \alpha^2 + 2 \sa \saa \alpha + 2 \alpha^2 \beta c_i + 2 \alpha \beta^2 c_i
%% \end{equation}}
\begin{equation}
  \nonumber \alpha >\frac{2 \sa \saa + 2 \beta ^2 c_i - 2 \beta \sa ^2- 2 c_i \sa^2 }{\sa ^2 + 2 \sa \beta - \saa - 2 \beta c_i}.
\end{equation}
We want $RHS<0$ and therefore want a $\beta$ such that
\[
\frac{2 \sa \saa + 2 \beta ^2 c_i - 2 \beta \sa ^2- 2 c_i \sa^2 }{\sa ^2 + 2 \sa \beta - \saa - 2 \beta c_i} \leq 0.
\]
As the denominator is always positive, we are only interested in the numerator, that is, finding a $\beta$ such that
\[
 \sa \saa +  \beta ^2 c_i -  \beta \sa ^2-  c_i \sa^2  \leq 0.
\]
This is satisfied for $\beta = \sa$
\[
 \sa \saa +  \sa ^2 c_i -   \sa ^3-  c_i \sa^2  \leq 0,
\]
which is clearly true.
\end{proof}
%$-{\tanh_{a,b}}$  & $-{\sum_j\tanh\left(\left|ac_j\right|^b\right)}$
\subsubsection{$-{\tanh_{a,b}}$  and \Da}\hfill
\label{sec:tanh2}
\begin{thm}
$-{\tanh_{a,b}}$  satisfies \[S(\bmat{cccccc}c_1& \cdots& c_i-\alpha& \cdots& c_j+\alpha& \cdots\emat ) < S(\vc),\] for
all $\alpha, c_i, c_j $ such that $c_i>c_j$ and $0<\alpha <
\frac{c_i-c_j}{2}$.
\end{thm}
 \begin{proof}
Need to show that
\begin{equation}
  \nonumber  - \tanh\left(a c_i-a \alpha\right)^b  - \tanh\left(a c_j+a \alpha\right)^b < -\tanh\left(a c_i\right)^b -\tanh\left(a c_j\right)
\end{equation}
Making the substitutions $x = a c_i$, $y = a c_j$ and $z = a\alpha$ we get
\begin{equation}
  \nonumber \tanh\left(x-z\right)^b +\tanh\left(y+z\right)^b> \tanh\left(x\right)^b + \tanh\left(y\right)^b
\end{equation}
with $x>y>0$ and $0<z<\frac{x-y}{2}$.  Setting
\begin{equation}
  \nonumber f(z) = \left(\tanh(x-z)^b-\tanh(x)^b\right) + \left(\tanh(y+z)^b -\tanh(y)^b\right),
\end{equation}
we use the mean value theorem of differential calculus to prove that
\begin{eqnarray*}
  \tanh(x-z)^b-\tanh(x)^b &=& -zb\left(1-\tanh^2\left(\theta_1)^b\right)\right)\\
  \tanh(y+z)^b-\tanh(y)^b &=& zb\left(1-\tanh^2\left(\theta_2)^b\right)\right)
\end{eqnarray*}
where $x-z<\theta_1<x$ and $y<\theta_2<y+z$.  However, because
$1-\tanh^2(x^b)$ is strictly decreasing for $x>0$ and $b>0$ because
$z<\frac{x-y}{2}\Leftrightarrow y+z<x-z$, it follows that
\begin{equation}
  \nonumber f(z)=zb\left[ \left(1-\tanh^2\left(\theta_2)^b\right)\right) -\left(1-\tanh^2\left(\theta_1)^b\right)\right)\right]>0.
\end{equation}
%% To satisfy this condition we must have $\frac{\partial S(c_1, \ldots,
%% c_i-\alpha, \ldots , c_j+\alpha, \ldots )}{\partial \alpha}< 0 $. We
%% show that this is not the case for this measure.
%% {\small
%% \begin{equation}
%%  %\frac{\partial S(c_1, \ldots, c_i-\alpha, \ldots , c_j+\alpha,\ldots )}{\partial \alpha}
%% \nonumber \begin{array}{r}\frac{\partial} {\partial \alpha}\left[ -\sum_{k\neq i,j} \tanh\left(a c_k\right)^b - \tanh\left(a c_i-a \alpha\right)^b - \tanh\left(a c_j+a \alpha\right)^b \right]\\
%% = (1- \tanh^2\left(a c_i-a \alpha\right)^b)b\left(a c_i-a \alpha\right)^{b-1}(-a)~~~~~~~~~~\\
%% + (1- \tanh^2\left(a c_j+a \alpha\right)^b)b\left(a c_i-a \alpha\right)^{b-1}(-a)~~~\end{array}
%% \end{equation}}
%% which is true as, for any $\theta$,
%% \begin{equation}
%%   \nonumber \tanh^2\theta<1
%% \end{equation}
%% and $c_i>\alpha$.
\end{proof}
\subsubsection{$-{\tanh_{a,b}}$  and \Dc }\hfill
\label{sec:tanh3}
\begin{thm}
$-{\tanh_{a,b}}$  satisfies \[S(\alpha + \vc)<S(\vc), ~\alpha \in \mathbb{R},~\alpha > 0.\]
\end{thm}
 \begin{proof}
It is enough to show that $\frac{\partial S(\alpha + \vc)}{\partial
\alpha}< 0 $ as if the derivative of the measure with respect to the
parameter $\alpha$ is negative then any $\alpha$ causes the measure to
drop.  \beq \nonumber \frac{\partial}{\partial\alpha} \left[- \sum_j
\tanh\left(a\alpha+ ac_j)^b\right)\right]~~~~~~~~~\\ \nonumber = - \sum_j
\left(1- \tanh^2\left((a\alpha+
c_j\alpha)^{b}\right)\right)b\left(a\alpha+ ac_j\right)^{b-1}a<0, \eeq
which is true as $a,b>0$ and $\tanh^2{\theta}<1$.
\end{proof}
% $ -\log$  & $-\sum_j\log\left(1+|c_j|^2\right)$ \\
\subsubsection{ $ -\log$  and \Dc  }\hfill
\label{sec:log3}
 \begin{proof} {\em  $ -\log$  satisfies  \[S(\alpha + \vc)<S(\vc), ~\alpha \in \mathbb{R},~\alpha > 0.\]}
as
\beq
&&-\sum_j \log\left(\frac{1+(\alpha +  c_j)^2)}{ (1+c_j^2 }\right) > 0
\eeq
Which is true because
\[
\iff  \frac{1 + (\alpha + c_j)^2}{1+  c_j^2} > 1,\alpha >0.
\]
\end{proof}

%$\kappa_4$ & $\frac{\sum_j |c_j|^4}{\left(\sum_j |c_j|^2\right)^2}$ \\
%% \subsubsection{ $\kappa_4$   and \Da  }\hfill
%% \label{sec:kur1}
%%  \begin{proof} {\em   $\kappa_4$  satisfies  \[S(\bmat{cccccc}c_1& \cdots& c_i-\alpha& \cdots& c_j+\alpha& \cdots\emat ) < S(\vc),\] for
%% all $\alpha, c_i, c_j $ such that $c_i>c_j$ and $0<\alpha <
%% \frac{c_i-c_j}{2}$.}\\
%% Again we take the derivative and show that for any $\alpha$ the
%% derivative of the measure decreases and hence the measure
%% decreases.
%% \[
%%  \frac{\partial}{\partial\alpha} \left[\frac{\sum_{k \neq i,j} |c_k|^4 +|c_i -\alpha|^4 +|c_j -\alpha|^4 }{\left(\sum_{k \neq i,j} |c_k|^2 +|c_i -\alpha|^2 +|c_j -\alpha|^2 \right)^2 }\right] >0.
%% \]
%% By \Ab~we may drop the absolute value in the measure without loss of generality: \beq
%%  \left[(c_j+\alpha)^3 - (c_i - \alpha)^3\right]  \left[(c_j+\alpha)^2 + (c_i - \alpha)^2\right] +\\
%% \left[(c_j+\alpha)^4 + (c_i - \alpha)^4\right] \left[ 2\alpha +c_j - c_i \right] <0,
%% \eeq
%% which we is true as
%% \beq
%% &\left[(c_j+\alpha) - (c_i - \alpha)\right]&<0\\
%% \left[ 2\alpha +c_j - c_i \right] <0 & \equiv&  \alpha < \frac{c_i-c_j}{2}.
%% \eeq
%% \end{proof}

%$\kappa_4$ & $\frac{\sum_j |c_j|^4}{\left(\sum_j |c_j|^2\right)^2}$ \\
\subsubsection{ $\kappa_4$   and \Db  }\hfill
\label{sec:kur2}
\begin{thm}
$\kappa_4$  satisfies  \[S(\alpha\vc) = S(\vc), ~\forall~ \alpha \in \mathbb{R},~\alpha > 0\]
\end{thm}
 \begin{proof}
\[
 \frac{\sum_j (\alpha c_j)^4}{\left(\sum_j (\alpha c_j)^2\right)^2} = \frac{\alpha^4 \sum_j c_j^4}{\alpha^4 \left(\sum_j c_j^2\right)^2} = \frac{\sum_j c_j^4}{\left(\sum_j c_j^2\right)^2}.
\]
\end{proof}
\subsubsection{ $\kappa_4$   and \Dc}\hfill
\label{sec:kur3}
\begin{thm}
$\kappa_4$  satisfies \[S(\alpha + \vc)<S(\vc), ~\alpha \in \mathbb{R},~\alpha > 0.\]
\end{thm}
\begin{proof}
Set
\begin{equation}
  \nonumber f(a) = \frac{\sum_i \left(c_i +\alpha\right)^4}{\left(\sum_i\left(c_i+\alpha\right)^2\right)^2}
\end{equation}
It follows that
\begin{small}
\begin{equation}
\nonumber \frac{\partial f}{\partial \alpha} = \frac{4\left[\sum_i\left(c_i+\alpha\right)^3 \sum_i\left(c_i+\alpha\right)^2 - \sum_i\left(c_i+\alpha\right)^4\sum_i\left(c_i+\alpha\right)\right]}{\left(\sum_i \left(c_i + \alpha\right)^2\right)^3}
\end{equation}\end{small}
%% \begin{small}
%% \begin{equation}
%% \nonumber \frac{\partial f}{\partial \alpha} = \frac{4\left(\sum_i\left(c_i+\alpha\right)^2\right)\left[\sum_i\left(c_i+\alpha\right)^3 \sum_i\left(c_i+\alpha\right)^2 - \sum_i\left(c_i+\alpha\right)^4\sum_i\left(c_i+\alpha\right)\right]}{\left(\sum_i \left(c_i + \alpha\right)^2\right)^4}
%% \end{equation}\end{small}
We can ignore the denominator as it is clearly positive. We claim that $\frac{\partial f}{\partial \alpha}<0$ for $\alpha>0$. This is because,  for positive $x_i$, it is always true that
\begin{equation}
  \nonumber \sum_i x_i^2 \sum_i x_i^3 < \sum_i x_i^4 \sum_i x_i
\end{equation}
as
\begin{eqnarray*}
 && \sum_i x_i^2 \sum_i x_i^3 - \sum_i x_i^4 \sum_i x_i \\~~~&=&  \sum_{i\neq j} \left(x_i^2 x_j^3 +x_i^3 x_j^2 - x_i^4 x_j - x_i x_j^4\right)\\~~~&=&  \sum_{i\neq j} x_i x_j\left[x_i x_j^2 +x_i^2 x_j - x_i^3 - x_j^3\right]\\~~~&=& -\sum_{i\neq j} x_i x_j\left(x_i- x_j\right)^2 \left(x_i+ x_j \right) <0.
\end{eqnarray*}

 %% \begin{proof} {\em   $\kappa_4$  satisfies  \[S(\alpha\vc) = S(\vc), ~\forall~ \alpha \in \mathbb{R},~\alpha > 0.\]}
%% \[
%%  \frac{\sum_j |\alpha c_j|^4}{\left(\sum_j |\alpha c_j|^2\right)^2} = \frac{\alpha^4 \sum_j |c_j|^4}{\alpha^4 \left(\sum_j |c_j|^2\right)^2} = \frac{\sum_j |c_j|^4}{\left(\sum_j |c_j|^2\right)^2}.
%% \]
\end{proof}

\subsubsection{ $\kappa_4$   and \De}\hfill
\label{sec:kur5}
\begin{thm}
$\kappa_4$  satisfies  $\forall i \exists \beta = \beta_i>0$ such that  $\forall \alpha >0:$ {\small \[ S(\bmat{cccc}c_1& \ldots& c_i+\beta+\alpha&\ldots\emat) > S(\bmat{cccc}c_1&\ldots&c_i+\beta& \ldots\emat).\]}
\end{thm}
 \begin{proof}
Fix $i$ and make the substitution $\tilde{c_i} = c_i +\beta$. We show
that the derivative of the measure is positive and hence the measure
increases for any $\alpha$ {\small
\begin{eqnarray*}
  \frac{\partial}{\partial \alpha} \left[\frac{\sum_{j \neq i}
{c}_j^4 + (\tilde{c}_i+\alpha)^4}{\left(\sum_{j \neq i}
{c}_j^2 + (\tilde{c}_i+\alpha)^2\right)^2}\right] &>&0.
\end{eqnarray*}}
The numerator of the derivative is
{\small
\begin{eqnarray*}
(\tilde{c}_i +\alpha)^3\left( \sum_{j \neq i} {c}_j^2 +
(\tilde{c}_i +\alpha)^2)\right) ~~~&&\\- \left(\sum_{j \neq i}{c}_k^4
+(\tilde{c}_i+\alpha)^4 \right)(\tilde{c}_i+\alpha) &>&0.
 \end{eqnarray*}}
Multiplying out and substituting back in for $\tilde{c_i}$ this
becomes
\[
 c_i +\alpha +\beta > \sqrt{\frac{\sum_{j \neq i} c_j^4 }{\left(\sum_{j \neq i} c_j^2 \right)}}.
\]
Clearly there exists a $\beta$ such that the above expression holds
true for all $\alpha>0$.
\end{proof}
% $\ell^p_{-}$ & $\sum_{j,c_j\neq 0}\left|c_j\right|^p,\,\,\, p<0$ \\
\subsubsection{$-\ell^p_{-}$ and \De}\hfill
\label{sec:lp_5}
\begin{thm}
$-\ell^p_{-}$   satisfies  $\forall i \exists \beta = \beta_i>0$, such that $\forall \alpha >0:$ {\small \[ S(\bmat{cccc}c_1& \ldots& c_i+\beta+\alpha&\ldots\emat) > S(\bmat{cccc}c_1&\ldots&c_i+\beta& \ldots\emat).\]}
\end{thm}
 \begin{proof}
Without loss of generality we can change the conditions slightly by
replacing $p$ ($p<0$) with $-p$ and correspondingly update the
constraint to $p>0$.
{\small
\begin{eqnarray*}
{\small  -\sum_{j\neq i,c_j\neq 0}c_j^{-p} -(c_i +\beta +\alpha)^{-p}}&>&{\small -\sum_{j\neq i,c_j\neq 0}c_j^{-p} -(c_i +\beta)^{-p}}\\
(c_i +\beta +\alpha)^{-p} &<&(c_i +\beta)^{-p}\\
\frac{1}{(c_i +\beta +\alpha)^{p}} &<&\frac{1}{(c_i +\beta)^{p}},
\end{eqnarray*}}
which is true if $\beta>0$.
\end{proof}
\subsubsection{$u_\theta$ and \Da}\hfill
\label{sec:u1}
\begin{thm}
$u_\theta$ does not satisfy  \[S(\bmat{cccccc}c_1& \cdots& c_i-\alpha& \cdots& c_j+\alpha& \cdots\emat ) < S(\vc),\] for
all $\alpha, c_i, c_j $ such that $c_i>c_j$ and $0<\alpha <
\frac{c_i-c_j}{2}$.
\end{thm}
 \begin{proof}
For $\theta = .5$,
\begin{eqnarray*}
  S\left([1,2,4,9]\right) &=&  .6667\\
S\left([1.1, 1.9, 4, 9]\right) &=&  .7333.
\end{eqnarray*}
The Robin Hood operation increased sparsity and hence does not satisfy
\Da.
\end{proof}
\subsubsection{$u_\theta$ and \Dc}\hfill
\label{sec:u3}
\begin{thm}
$u_\theta$ does not satisfy  \[S(\alpha + \vc)<S(\vc), ~\alpha \in \mathbb{R},~\alpha > 0.\]
\end{thm}
 \begin{proof}
The support of $\vc$ is $[c_{(1)},c_{(N)}]$. Assume the support of the
$\lceil \theta N \rceil$ points that correspond to the minimum  is $[c_{(k)},c_{(j)}]$.  By adding a constant, $\alpha$, to each
coefficient in the distribution we shift the distribution to $\vc+\alpha$. Clearly, neither of the two supports mentioned above changes:
 $(c_{(j)}-\alpha)-(c_{(k)}-\alpha)
=c_{(j)}-c_{(k)}$. Hence $u_\theta$ does not
satisfy \Dc.
\end{proof}
\subsubsection{$u_\theta$ and \Dd}\hfill
\label{sec:u4}
\begin{thm}
$u_\theta$ satisfies  \[S(\vc) = S(\vc\|\vc) = S(\vc\|\vc\|\vc) = S(\vc\|\vc\|\cdots\|\vc).\]
\end{thm}
 \begin{proof}
The support of $\vc$ is $[c_{(1)},c_{(N)}]$. Assume the support of the
$\lceil \theta N \rceil$ points that correspond to the minimum  is $[c_{(k)},c_{(j)}]$. The new set, $\{\vc \| \vc\}$ has $2\lceil N\theta\rceil$
points lying between values $c_{(j)}$ and $c_{(k)}$, that is, neither of the previously mentioned two supports has changed. This reasoning holds for
cloning the data more than once. Hence $u_\theta$ satisfies \Dd.
\end{proof}
\subsubsection{$u_\theta$ and \De}\hfill
\label{sec:u5}
\begin{thm}
$u_\theta$ satisfies $\forall i \exists \beta = \beta_i>0$, such that $\forall \alpha >0: $ {\small \[S(\bmat{cccc}c_1& \ldots& c_i+\beta+\alpha&\ldots\emat) > S(\bmat{cccc}c_1&\ldots&c_i+\beta& \ldots\emat).\]}
\end{thm}
 \begin{proof}
 The support of $\vc$ is $[c_{(1)},c_{(N)}]$. Without loss of generality we focus on $c_{(N)}$ as the effect of adding sufficiently large $\beta$ to any other coefficient will result in this coefficient becoming the largest. We choose $\beta$ sufficiently large so that $c_{(N)}+\beta$ is set sufficiently far apart from the other coefficients for
the support of the $\lceil \theta N \rceil$ points that correspond to the minimum not to contain $c_{(N)}$. Consequently,
 the numerator of the minimization term is a constant $K$ not depending on $\beta$ or $\alpha$. We can rewrite
{\small \[
S(\bmat{cccc}c_1& \ldots& c_i+\beta+\alpha&\ldots\emat) > S(\bmat{cccc}c_1&\ldots&c_i+\beta& \ldots\emat)
\]}
as
\[
1-\frac{K}{c_{(N)}-c_{(1)}+\alpha +\beta}<1-\frac{K}{c_{(N)}-c_{(1)}+\alpha}
\]
which is clearly true and the proof is complete.
%The support of $\vc$ is $[c_{(1)},c_{(N)}]$. Assume the support of the
%$\lceil \theta N \rceil$ points that correspond to the minimum  is $[c_{(k)},c_{(j)}]$. By
%increasing $c_{(N)}$ the range of the $\lceil \theta N\rceil$ points
%of interest is unaffected (note that, by definition, $c_N$ cannot appear in the numerator)
%but the length of $[c_{(1)},c_{(N)}]$ increases, and this clearly increases the sparsity measure.
%
% If we wish to do the same with a smaller coefficient we are faced
%with two scenarios. The first scenario is that the coefficient of
%interest, $c_{(i)}$ is not in the range $c_{(j)}-c_{(k)}$ in which
%case there exists a $\beta$ which will change the sorted ordering to
%make $c_{(i)}$ the largest coefficient and the first argument holds
%true and the condition is satisfied. The second scenario is that the
%coefficient of interest is in the range $c_{(j)}-c_{(k)}$. In this
%case there exists a $\beta$ which makes the coefficient of interest
%larger than $c_{(N)}$. A new range is defined, $c_{(j+1)}-c_{(k)}$ or
%$c_{(j)}-c_{(k-1)}$ which may initially decrease sparsity but for a
%sufficiently large $\beta$ will increase sparsity and hence satisfy
%the condition.
\end{proof}
\subsubsection{$u_\theta$ and \Df}\hfill
\label{sec:u6}
\begin{thm}
$u_\theta$ does not  satisfy  \[S(\vc||0) > S(\vc).\]
\end{thm}
 \begin{proof}
Assume $\vc$ has total support $c_{(N)}-c_{(1)}$ and the support of
$\lceil \theta N\rceil $ points lying between values
$c_{(j)}-c_{(k)}$.  If $0$ lies within the range $c_{(j)}-c_{(k)}$
adding a $0$ will decrease the range to $c_{(j-1)}-c_{(k)}$ without
increasing the total support.
%However, if $0$ is outside the
%range $c_{(j)}-c_{(k)}$ adding a $0$ will increase the total support
%without affecting the range $c_{(j)}-c_{(k)}$ and will result in a
%decrease in sparsity and hence does not satisfy \Df.
\end{proof}
% $H_G$  & $-\sum_j\ln \left|c_j\right|^2$ \\
\subsubsection{$H_G$    and \Da}\hfill
\label{sec:hg1}
\begin{thm}
 $H_G$   satisfies \[S(\bmat{cccccc}c_1& \cdots& c_i-\alpha& \cdots& c_j+\alpha& \cdots\emat ) < S(\vc),\] for
all $\alpha, c_i, c_j $ such that $c_i>c_j$ and $0<\alpha <
\frac{c_i-c_j}{2}$.
\end{thm}
 \begin{proof}
 \beq
&-\sum_{k \neq i,j}\ln c_k^2 - \ln\left(c_i-\alpha\right)^2 - \ln \left(c_j+\alpha\right)^2 <-\sum_{k }\ln c_k^2&\\
&-2\ln\left(c_i-\alpha\right) -2\ln\left(c_j+\alpha\right)< -2\ln c_i -2\ln c_j&\\
& \left( c_i - \alpha\right)\left(c_j +\alpha\right)> c_i c_j&\\
& a<c_i-c_j,&
\eeq
which is clearly true.
\end{proof}
\subsubsection{ Hoyer   and \Da}\hfill
\label{sec:hoy1}
\begin{thm}
Hoyer  satisfies \[S(\bmat{cccccc}c_1& \cdots& c_i-\alpha& \cdots& c_j+\alpha& \cdots\emat ) < S(\vc),\] for
all $\alpha, c_i, c_j $ such that $c_i>c_j$ and $0<\alpha <
\frac{c_i-c_j}{2}$.
\end{thm}
 \begin{proof}
{\small
\begin{multline*}
\frac{\partial}{\partial \alpha}\frac { (\sqrt{N}-\frac{\ell^1}{\ell^2})}{(\sqrt{N}-1)} \\
\equiv \frac{\partial}{\partial \alpha} \left[\frac{-1}{\sqrt{N}-1}\left(\frac{\sum_j c_j}{\left( \sum_{k \neq i,j} c_k^2 +(c_i -\alpha)^2 +(c_j + \alpha)^2\right)^{\frac{1}{2}}}\right)\right],
\end{multline*}}
which is
{\small
\begin{multline}
\nonumber {\textstyle \frac{\sum_j c_j}{\sqrt{N}-1}} \left(\sum_{k \neq i,j} \left( c_k^2 +(c_i -\alpha)^2 +(c_j + \alpha)^2\right)^{-\frac{3}{2}}\right) (c_j -c_i-2\alpha)\\ <0.
\end{multline}}
This is true as $(c_j -c_i-2\alpha)<0$.
\end{proof}
\subsubsection{Hoyer   and \Dc}\hfill
\label{sec:hoy3}
\begin{thm}
Hoyer satisfies \[S(\alpha + \vc)<S(\vc), ~\alpha \in \mathbb{R},~\alpha > 0.\]
\end{thm}
 \begin{proof}
{\small
\begin{eqnarray*}
&&\frac{\partial}{\partial \alpha}\left( \frac { \sqrt{N}-\frac{\sum_{i=1}^N (c_i+\alpha)}{\sqrt{\sum_{i=1}^N  (c_i+\alpha)^2}}}{(\sqrt{N}-1)}\right)
 \\&&~~~~~\equiv
 \frac{\partial}{\partial \alpha} \left[\frac{-1}{\sqrt{N}-1}\left(\sum_{i=1}^N  (c_i +\alpha)\right) \left(\sum_{i=1}^N   (c_i +\alpha)^2\right)^{-\frac{1}{2}}\right].
\end{eqnarray*}}
With the substitution
\begin{eqnarray*}
  \sa &=& \sum_{i=1}^N c_i\\
\saa &=& \sum_{i=1}^N c_i^2
\end{eqnarray*}
this becomes
\[
(\sa+N\alpha)^2 (\saa +2 \alpha \sa + N\alpha^2)^{-\frac{3}{2}} - N( \saa + 2 \alpha \sa + N \alpha^2)<0,
\]
which  simplifies to
\begin{equation}
  \nonumber N>\frac{\sa^2}{\saa}.
\end{equation}
We rewrite this as
{\small \begin{equation}
  \nonumber N\saa= \sum_{i=1}^N 1 \sum_{i=1}^Nc_i^2 > \left(\sum_{i = 1}^N c_i\right)^2 = \sa,
\end{equation}}
which is true by Cauchy-Schwarz.
\end{proof}
%  ${ (\sqrt{N}-\frac{\ell^1}{\ell^2})}{(\sqrt{N}-1)}^{-1}$ & $(\sqrt{N}-\frac{\sum_j |c_j|}{\sqrt{\sum_j |c_j|^2}})(\sqrt{N}-1)^{-1}$ \\
\subsubsection{ Hoyer and \De}\hfill
\label{sec:hoy5}
\begin{thm}
 Hoyer satisfies  $\forall i \exists \beta = \beta_i>0$, such that $\forall \alpha >0:$ {\small \[S(\bmat{cccc}c_1& \ldots& c_i+\beta+\alpha&\ldots\emat) > S(\bmat{cccc}c_1&\ldots&c_i+\beta& \ldots\emat).\]}
\end{thm}
 \begin{proof}
{\small
\begin{equation}
\nonumber \frac{\partial}{\partial \alpha}\frac { (\sqrt{N}-\frac{\ell^1}{\ell^2})}{(\sqrt{N}-1)}  \equiv \frac{\partial}{\partial \alpha} \left[\frac{-1}{\sqrt{N}-1}\left(\frac{\sum c_j +\alpha +\beta}{\sum_{k \neq i} \left( c_k^2 +(c_i +\alpha +\beta)^2 \right)^{\frac{1}{2}}}\right)\right],
\end{equation}}
which is
{\small
\begin{equation}
\nonumber
-\frac{\left(\sum_{j\neq i}c_j^2 +\left(c_i +\alpha + \beta\right)^2\right)^\frac{3}{2}}{\sqrt{N} -1}
\left[ \left(\sum_{j\neq i} c_j\right)\left(c_j +\alpha +\beta\right)- \sum_{j\neq i}c_j^2 \right]
\end{equation}}
Clearly for sufficiently large $\beta$ the above quantity is $>0$.
\end{proof}
\subsubsection{ Gini and \Da}\hfill
\label{sec:gini1}
\begin{thm}
The Gini Index  satisfies
\[S(c_1, \ldots, c_i-\alpha, \ldots , c_j+\alpha, \ldots ) < S(c),\] for
all $\alpha , c_i, c_j $ such that $c_i>c_j$ and $0<\alpha <
\frac{c_i-c_j}{2}$.
\end{thm}
 \begin{proof}
The Gini Index of $\vc=\bmat{cccc} c_1& c_2&c_3&\cdots \emat$ is given by
\begin{equation}
S(\vc) = 1 - 2\sum_{k=1}^N \frac{c_{(k)}}{\|\vc\|_1}\left(\frac{N-k+\frac{1}{2}}{N}  \right),
\end{equation}
where $(k)$ denotes the new index after sorting from lowest to highest, that
is, $c_{(1)}\leq c_{(2)} \leq\cdots\leq c_{(N)}$.

Without loss of generality we can assume that the two coefficients
involved in the Robin Hood operation are $c_{(i)}$ and $c_{(j)}$.
After a Robin Hood operation is performed on $\vc$ we label the
resulting set of coefficients $\vd$ which are sorted using an index
which we denote $[\cdot]$, that is, $d_{[1]}\leq d_{[2]}
\leq\cdots\leq d_{[N]}$.  Let us assume that the Robin Hood operation
alters the sorted ordering in that the new coefficient obtained by the
subtraction of $\alpha$ from $c_{(i)}$ has the new rank $i-n$, that
is,
\begin{equation}
  \nonumber d_{[i-n]} = c_{(i)}-\alpha
\end{equation}
and the new coefficient obtained by the addition of $\alpha$ to
$c_{(j)}$ has the new rank $j+m$, that is,
\begin{equation}
  \nonumber d_{[j+m]} = c_{(j)}+\alpha .
\end{equation}
The correspondence between the coefficients of $\vc$ and $\vd$ is
shown in Fig~\ref{fig:mapping}
%RESUME HERE
%% {\small
%% \xymatrix@R=6cm@C=1cm{
%% c_{(1)} \ar[dd]& c_{(2)}\ar[dd] & \cdots & c_{(j-1)}\ar[dd]&c_{(j)}
%% \ar@{--}[d] \ar@{--}[d];[dr] \ar@{}[dr];[drr]|{\cdots} \ar@{--}[drr];[drrr] \ar@{-->}[drrr];[drrrd]
%% &c_{(j+1)}\ar[ddl] &\cdots&c_{(j+m)}\ar[ddl] &c_{(j+m+1)}\ar[dd]&\cdots\\
%% &&&&&&&&&\\
%% d_{[1]}& d_{[2]} & \cdots & d_{[j-1]} & d_{[j]}&  d_{[j+1]} &\cdots&d_{[j+m]}&d_{[j+m+1]}&\cdots
%%  }}

%% {\small
%% \xymatrix@=.5pt@C=10pt{
%% \cdots&c_{(i-n-1)} \ar[dd]& c_{(i-n)}\ar[ddr] & \cdots & c_{(i-1)}\ar[ddr]&c_{(i)}
%% \ar@{--}[d] \ar@{--}[d];[dl] \ar@{}[dl];[dll]|{\cdots} \ar@{--}[dll];[dlll] \ar@{-->}[dlll];[dllld]
%% &c_{(i+1)}\ar[dd] &\cdots&c_{(N)}\ar[dd] \\
%% &&&&&&&&&\\
%% \cdots&d_{[i-n-1]}& d_{[i-n]}&d_{[i-n+1]} & \cdots &  d_{[i]}&  d_{[i+1]} &\cdots&d_{[N]}
%% }}
\begin{figure*}
{\small
\xymatrix@C=.5pt@R=10pt{
c_{(1)} \ar[dd]& c_{(2)}\ar[dd] & \cdots & c_{(j-1)}\ar[dd]&c_{(j)}
\ar@{--}[d] \ar@{--}[d];[dr] \ar@{}[dr];[drr]|{\cdots} \ar@{--}[drr];[drrr] \ar@{-->}[drrr];[drrrd]
&c_{(j+1)}\ar[ddl] &\cdots&c_{(j+m)}\ar[ddl] &c_{(j+m+1)}\ar[dd]&\cdots&c_{(i-n-1)} \ar[dd]& c_{(i-n)}\ar[ddr] & \cdots & c_{(i-1)}\ar[ddr]&c_{(i)}
\ar@{--}[d] \ar@{--}[d];[dl] \ar@{}[dl];[dll]|{\cdots} \ar@{--}[dll];[dlll] \ar@{-->}[dlll];[dllld]
&c_{(i+1)}\ar[dd] &\cdots&c_{(N)}\ar[dd]\\
&&&&&&&&&&&&&&&&\\
d_{[1]}& d_{[2]} & \cdots & d_{[j-1]} & d_{[j]}&  \cdots& d_{[j+m-1]} &d_{[j+m]}&d_{[j+m+1]} &\cdots&d_{[i-n-1]}& d_{[i-n]}&d_{[i-n+1]} & \cdots &  d_{[i]}&  d_{[i+1]} &\cdots&d_{[N]}
 }
}\caption{The mapping between a vector before and after a Robin Hood operation. This is used in Proof~\ref{sec:gini1}}\label{fig:mapping}
\end{figure*}
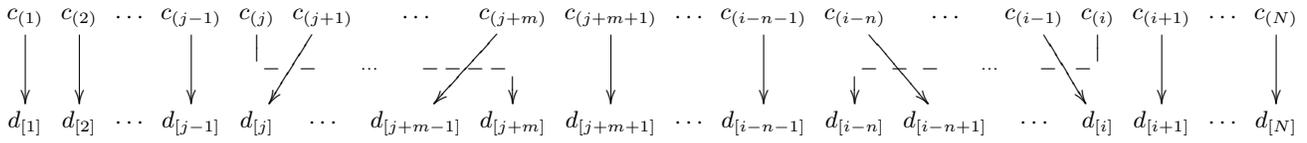
%% \begin{center}
%%   \includegraphics[width=12cm]{gini_D1_mapping}
%% \end{center}
and in mathematical terms is
{\small \begin{eqnarray*}
 d_{[k]} = c_{(k)}  & \text{~for~} &  1\leq k \leq j-1 \\
 d_{[k]} = c_{(k+1)}  & \text{~for~} &  j\leq k \leq j+m-1 \\
 d_{[k]} = c_{(j)}+\alpha  &  \text{~for~} &   k = j+m \\
 d_{[k]} = c_{(k)}  &  \text{~for~} &  j+m+1\leq k \leq i-n-1 \\
 d_{[k]} = c_{(i)}-\alpha  &  \text{~for~} &   k = i-n \\
 d_{[k]} = c_{(k-1)}  &  \text{~for~} &  i-n+1\leq k \leq i \\
 d_{[k]} = c_{(k)}  &  \text{~for~} &  i+1\leq k \leq N.
\end{eqnarray*}}
We wish to show
\begin{equation}
  \nonumber S(\vc)> S(\vd)
\end{equation}
Removing common terms and noting that $\|\vc\|_1 = \|\vd\|_1$ we can
simplify this to
\begin{equation}
  \nonumber \sum_{k \in \Delta} c_{(k)}\left(N-k+\half\right) < \sum_{k \in \Delta} d_{[k]}\left(N-k+\half\right),
\end{equation}
where $\Delta = \{ j, j+1, \ldots,j+m,i-n, i-n+1, \ldots,i\}$.
Using the correspondence above we can express the coefficients of
$\vd$ in terms of the coefficients of $\vc$. We then get
\begin{eqnarray*}
  \begin{array}{c}
   \sum_{k=1}^m c_{(j+k)}\left[\left(N-j-k+1+\half\right) -\left(N-j-k+\half\right)\right] \\
 +  \sum_{k=1}^n c_{(i-k)}\left[\left(N-i+k-1+\half\right) -\left(N-i+k+\half\right)\right] \\
  +  c_{(j)}\left[\left(N-j-m+\half\right) -\left(N-j+\half\right)\right] \\
  +  c_{(i)}\left[\left(N-i+n+\half\right) -\left(N-i+\half\right)\right] \\
  + \alpha \left[\left(N-j-m+\half\right) -\left(N-i+n+\half\right)\right]>0,
  \end{array}
\end{eqnarray*}
which becomes
{\small
\begin{eqnarray*}
  %\sum_{k=1}^m c_{(j+k)} - \sum_{k=1}^n c_{(i-k)} - m c_{(j)} +n c_{(i)} + \alpha \left( i-n - (j+m)\right) &>&0\\
% \sum_{k=1}^m \left(c_{(j+k)}-c_{(j)}\right) + \sum_{k=1}^n \left(c_{(i)} - c_{(i-k)}\right)  + \alpha \left( (i-n) - (j+m)\right) &>&0
&& \sum_{k=1}^m \left(c_{(j+k)}-c_{(j)}\right) + \sum_{k=1}^n \left(c_{(i)} - c_{(i-k)}\right)\\&&~~~  + \alpha \left( (i-n) - (j+m)\right) >0.
\end{eqnarray*}}
This is true as the two summations are positive as the negative
component has a lower sorted index than the positive and is hence
smaller and the last term is positive due to the condition on
$\alpha$.
\end{proof}
\subsubsection{ Gini and \Db  }\hfill
\label{sec:gini2}
\begin{thm}
The Gini Index  satisfies  \[S(\alpha\vc) = S(\vc), ~\forall~ \alpha \in \mathbb{R},~\alpha > 0.\]
\end{thm}
 \begin{proof}
\begin{eqnarray*}
  S(\alpha \vc)& =& 1 - 2\sum_{k=1}^N \frac{\alpha c_{(k)}}{\|\alpha \vc\|_1}\left(\frac{N-k+\frac{1}{2}}{N}  \right)\\
&=& 1 - 2\sum_{k=1}^N \frac{\alpha c_{(k)}}{\alpha\|\vc\|_1}\left(\frac{N-k+\frac{1}{2}}{N}  \right)\\
&=& S(\vc).
\end{eqnarray*}
\end{proof}

\subsubsection{ Gini   and \Dc}\hfill
\label{sec:gini3}
\begin{thm}
 The Gini Index satisfies \[S(\alpha + \vc)<S(\vc), ~\alpha \in \mathbb{R},~\alpha > 0.\]
\end{thm}
 \begin{proof}
Rewriting $S(\alpha +\vc)< S(\vc)$ and
%% \begin{equation}
%% \nonumber 1-2 \sum_{k=1}^N\frac{c_{(k)}+\alpha}{\|\vc +\alpha\|_1}\left(\frac{N-k+\half}{N}\right) < 1-2 \sum_{k=1}^N \frac{c_{(k)}}{\|\vc\|_1}\left(\frac{N-k+\half}{N}\right)
%% \end{equation}
making the substitution
\begin{equation}
\nonumber f(k) = \left(\frac{N-k+\half}{N}  \right),
\end{equation}
we get the following:
{\small \begin{eqnarray*}
\sum_{k=1}^N \frac{c_{(k)}}{\|\vc+\alpha\|_1}f(k)  +
 \frac{N\alpha}{\|\vc+\alpha\|_1}\sum_{k=1}^N f(k)  -\sum_{k=1}^N \frac{c_{(k)} }{\|\vc\|_1}f(k)  & >& 0\\
\sum_{k=1}^N c_{(k)} f(k)  \left( \frac{1}{\|\vc+\alpha \|_1} - \frac{1}{\|\vc\|_1}\right) +
 \frac{N\alpha}{\|\vc+\alpha \|_1}\sum_{k=1}^Nf(k)  & >& 0\\
\sum_{k=1}^N \frac{c_{(k)}}{\|\vc\|_1}
 f(k)  \left( \frac{-N\alpha}{\|\vc+\alpha\|_1}\right) +
 \frac{N\alpha}{\|\vc+\alpha\|_1}\sum_{k=1}^N f(k) & >& 0\\
\sum_{k = 1}^N f(k) \left( 1- \frac{c_{(k)}}{\|\vc\|_1}\right)&>&0.
\end{eqnarray*}}
This is clearly true for $N>1$.
\end{proof}
\subsubsection{ Gini   and \Dd }\hfill
\label{sec:gini4}
\begin{thm}
The Gini Index satisfies \[S(\vc) = S(\vc\|\vc) = S(\vc\|\vc\|\vc) = S(\vc\|\vc\|\cdots\|\vc).\]
\end{thm}
 \begin{proof}
We clone $\vc$ $M$ times to get the vector $\vd$ which has length $MN$:
\begin{eqnarray*}
  &&\begin{array}{ccc}S(&\underbrace{\vc\|\cdots\|\vc}&) =S(\vd)\\&M&\end{array}\\
&=& 1 - 2 \sum_{k=1}^{MN} \frac{d_{(k)}}{\|\vd\|_1}\left(\frac{MN-k +\half}{MN}\right)\\
&=& 1 - 2 \sum_{j=1}^{M} \sum_{i=1}^{N} \frac{c_{(k)}}{M\|\vc\|_1}\left(\frac{MN-(Mi-M+j) +\half}{MN}\right)\\
&=& 1 - 2 \sum_{i=1}^{N} \frac{c_{(k)}}{\|\vc\|_1} \sum_{j=1}^{M}\left(\frac{MN-Mi+M-j +\half}{M^2N}\right)\\
&=& 1 - 2 \sum_{i=1}^{N} \frac{c_{(k)}}{\|\vc\|_1} \left(\frac{M^2N-M^2i+M^2-\frac{M(M+1)}{2} +\frac{M}{2}}{M^2N}\right)\\
&=& 1 - 2 \sum_{i=1}^{N} \frac{c_{(k)}}{\|\vc\|_1} \left(\frac{M^2N-M^2i+M^2-\frac{M^2}{2}-\frac{M}{2} +\frac{M}{2}}{M^2N}\right)\\
&=& 1 - 2 \sum_{i=1}^{N} \frac{c_{(k)}}{\|\vc\|_1} \left(\frac{M^2N-M^2i+\frac{M^2}{2}}{M^2N}\right)\\
&=& 1 - 2 \sum_{i=1}^{N} \frac{c_{(k)}}{\|\vc\|_1} \left(\frac{N-i+\half}{N}\right)\\
&=& S(\vc).
\end{eqnarray*}
\end{proof}
\subsubsection{ Gini   and \De}\hfill
\label{sec:gini5}
\begin{thm}
The Gini Index satisfies $\forall i \exists \beta = \beta_i>0$, such that $\forall \alpha >0:$ {\small \[S(\bmat{cccc}c_1& \ldots& c_i+\beta+\alpha&\ldots\emat) > S(\bmat{cccc}c_1&\ldots&c_i+\beta& \ldots\emat).\]}
\end{thm}
 \begin{proof}
We use the following notation,
\begin{eqnarray*}
  \vc& = &\{ c_{(1)}, c_{(2)}, \ldots, c_{(N)}+\beta\}.
\end{eqnarray*}
Without loss of generality we have chosen to perform the operation on $c_{(N)}$
as $\beta$ can absorb the additive value needed to change any of the $c_{(i)}$
to $c_{(N)}$.

We wish to show that
{\footnotesize
\begin{eqnarray*}
  &&1-2\sum_{i=1}^{N} \frac{c_{(i)}}{\|\vc\|_1}\left(\frac{N-i+\half}{N}\right) \\&&~~~<
1-2\sum_{i=1}^{N} \frac{c_{(i)}}{\|\vc\|_1+\beta}\left(\frac{N-i+\half}{N}\right) - \frac{\beta}{N(\|\vc\|_1+\beta)}.
\end{eqnarray*}}
We can simplify the above to
{\footnotesize
\begin{eqnarray*}
\sum_{i=1}^{N} c_{(i)} \left(\frac{N-i+\half}{N}\right)\left(\frac{1}{\|\vc\|_1} - \frac{1}{\|\vc\|_1+\beta} \right) &>& \frac{\beta}{2N(\|\vc\|_1+\beta)}\\
\sum_{i=1}^{N} c_{(i)} \left(N-i+\half \right) &>& \frac{\|\vc\|_1}{2}=
\frac{1}{2}\sum_{i=1}^N{} c_{(i)} \\
\sum_{i=1}^{N} c_{(i)} \left(N-i\right) &>& 0.
\end{eqnarray*}
}
Hence, the Gini Index satisfies \De.
\end{proof}
\subsubsection{ Gini   and \Df   }\hfill
\label{sec:gini6}
\begin{thm}
The Gini Index satisfies satisfy \[S(\vc||0) > S(\vc).\]
\end{thm}
 \begin{proof}
 Let us define
\begin{equation}
  \nonumber \vd = \vc||0 = \bmat{cccccc} c_1& c_2&c_3&\cdots&c_N&0 \emat
\end{equation}
and we note that $\|\vd\|_1 = \|\vc\|_1$. Without loss of generality
we assign the lowest rank to the added coefficient $0$, that is,
$d_{N+1} = d_{(1)}$. We can now make the assertion $d_{(i+1)} =
c_{(i)}$, yielding
\begin{eqnarray*}
\nonumber S(\vd)=&&1 -2\sum_{k=2}^{N+1}\frac{d_{(k)}}{|\vd|}\left(\frac{N+1-k+\half}{N+1}\right)\\&&~~~ -2\frac{0}{|\vd|}\left(\frac{N+1-1+\half}{N+1}\right).
\end{eqnarray*}
Making the substitution $i = k-1$ we get
\begin{eqnarray*}
   S(\vd)&=&1 -2\sum_{i=1}^{N}\frac{d_{(i+1)}}{|\vd|}\left(\frac{N+1-i+\half}{N+1}\right)\\
 &=&1 -2\sum_{i=1}^{N}\frac{c_{(i)}}{\|\vc\|_1}\left(\frac{N-i+\half}{N+1}\right)\\
&>&1 -2\sum_{i=1}^{N}\frac{c_{(i)}}{\|\vc\|_1}\left(\frac{N-i+\half}{N}\right) \\
&=& S(\vc).
\end{eqnarray*}
\end{proof}
%\end{onecolumn}
%\twocolumn

%%------------------------------------------------------%%
%%------------------------------------------------------%%
%%------------------------------------------------------%%
%\addcontentsline{toc}{section}{References}
{%\footnotesize
\bibliography{sparse_journal_arXiv}
%\bibliography{c:/references_niall}
\bibliographystyle{IEEEtran}
}

%% \begin{IEEEbiography}{Michael Shell}
%% Biography text here.
%% \end{IEEEbiography}

%% % if you will not have a photo at all:
%% \begin{IEEEbiographynophoto}{John Doe}
%% Biography text here.
%% \end{IEEEbiographynophoto}

% that's all folks
\end{document}